\documentclass[journal,onecolumn]{IEEEtran}
\usepackage{amsmath}
\usepackage{booktabs}
\usepackage{latexsym}
\usepackage{mathrsfs}
\usepackage{amsthm}
\usepackage{amssymb}
\usepackage{amsfonts}
\usepackage{amsbsy}

\usepackage[shortlabels]{enumitem}
\usepackage{url}
\usepackage{array}
\usepackage{pdflscape}
\usepackage{xcolor}
\usepackage{stmaryrd}
\usepackage{verbatim}
\usepackage{braket}

\newcommand{\Ff}{{\mathbb F}}

\newcommand\rank{\operatorname{rank}}   
\newcommand\cc{{\mathcal C}}        %

\newcommand{\F}{ {{\mathbb F}} }

\DeclareMathOperator{\Res}{Res}
\DeclareMathOperator{\Tr}{Tr}
\DeclareMathOperator{\Gal}{Gal}
\DeclareMathOperator{\supp}{supp}
\DeclareMathOperator{\diag}{diag}

\theoremstyle{plain}
\newtheorem{thm}{Theorem}
\newtheorem{lem}[thm]{Lemma}
\newtheorem{prop}[thm]{Proposition}

\newtheorem{cor}{Corollary}
\theoremstyle{definition}

\newtheorem{example}{Example}
\newtheorem{remark}{Remark}
\usepackage[colorlinks=true]{hyperref}

\begin{document}
\title{On Linear Codes Whose Hermitian Hulls are MDS}

\author{Gaojun Luo, Lin Sok, Martianus Frederic Ezerman, and San Ling
\thanks{G. Luo, Lin Sok, M. F. Ezerman, and S. Ling are with the School of Physical and Mathematical Sciences, Nanyang Technological University, 21 Nanyang Link, Singapore 637371, e-mails: $\{\rm gaojun.luo, fredezerman, lin.sok, lingsan\}$@ntu.edu.sg.}
\thanks{G. Luo, L. Sok, M. F. Ezerman, and S. Ling are supported by Nanyang Technological University Research Grant No. 04INS000047C230GRT01.}

}


\maketitle

\begin{abstract}
Hermitian hulls of linear codes are interesting for theoretical and practical reasons alike. In terms of recent application, linear codes whose hulls meet certain conditions have been utilized as ingredients to construct entanglement-assisted quantum error correcting codes. This family of quantum codes is often seen as a generalization of quantum stabilizer codes. Theoretically, compared with the Euclidean setup, the Hermitian case is much harder to deal with. Hermitian hulls of MDS linear codes with low dimensions have been explored, mostly from generalized Reed-Solomon codes. Characterizing Hermitian hulls which themselves are MDS appears to be more involved and has not been extensively studied.

This paper introduces some tools to study linear codes whose Hermitian hulls are MDS. Using the tools, we then propose explicit constructions of such codes. We consider Hermitian hulls of both Reed-Solomon and non Reed-Solomon types of linear MDS codes. We demonstrate that, given the same Hermitian hull dimensions, the codes from our constructions have dimensions which are larger than those in the literature.
\end{abstract}


\section{Introduction}\label{sec:intro}

Let $q$ be a prime power and let $\Ff_q$ be the finite field with $q$ elements. An $[n,k,d]_{q^2}$ linear code $\cc$ is a $k$-dimensional subspace of $\Ff_{q^2}^n$ with minimum (Hamming) distance $d$. Such a code is \emph{maximum distance separable} (MDS) if $d=n-k+1$. Linear codes $\cc_1$ and $\cc_2$ over $\Ff_q$ are \emph{monomially-equivalent} if $\cc_1$ can be transformed into $\cc_2$ by a permutation of the coordinates of  each codeword in $\cc_1$ and the multiplication of every coordinate of the codeword by a nonzero element in $\Ff_q$. Let $S$ be a subset of $[n]$. The \emph{punctured code} of $\cc$ with respect to $S$ is the code obtained by deleting the entries indexed by the set $S$ in each codeword of $\cc$.

The intersection of a linear code and its dual code has become a major area of interest in coding theory. Given an inner product, typically Euclidean or Hermitian, such an intersection set may have an important role in addressing code equivalence \cite{Sendrier2000}, orthogonal direct sum masking \cite{Bringer2014,Carlet2016}, and quantum error correction \cite{Calderbank1998,Poulin2005,Brun2006}.

The \emph{Hermitian inner product} of vectors $\mathbf{u}=(u_1,\ldots,u_n), \mathbf{v}=(v_1,\ldots,v_n) \in \Ff_{q^2}^n$ is $\langle\mathbf{u},\mathbf{v}
\rangle_{{\rm H}}=\sum_{i=1}^nu_iv_i^q$. The \emph{Hermitian dual} of $\cc$ is
\[
\mathcal{C}^{\perp_{\rm H}} = \left\{\mathbf{u} \in \mathbb{F}_{q^2}^n : \langle\mathbf{u},\mathbf{v}\rangle_{{\rm H}}=0,
\mbox{ for all } \mathbf{v} \in \mathcal{C} \right\}.
\]
The intersection $\cc\cap\mathcal{C}^{\perp_{\rm H}}$, denoted by ${\rm Hull}_{\rm H}(\cc)$, is the \emph{Hermitian hull} of $\cc$. The notion of hull was introduced in \cite{Assmus1990} to classify finite projective planes. There are two extremal cases. When ${\rm Hull}_{\rm H}(\cc)=\cc$ or $\mathcal{C}^{\perp_{\rm H}}$, then $\cc$ is \emph{Hermitian self-orthogonal} or \emph{Hermitian dual-containing}, respectively. A code $\cc$ is \emph{Hermitian linear complementary dual} (LCD) if ${\rm Hull}_{\rm H}(\cc) = \{\mathbf{0}\}$. The \emph{Euclidean inner product} of vectors $\mathbf{u}=(u_1,\ldots,u_n), \mathbf{v}=(v_1,\ldots,v_n) \in \Ff_{q}^n$ is $\langle\mathbf{u},\mathbf{v}\rangle_{{\rm E}}=\sum_{i=1}^nu_iv_i$. Analogously, duality and hull are applicable in the Euclidean inner product setup, with the notation involving ${\rm H}$ changed into ${\rm E}$.

\subsection{Known Results on Hulls}

Euclidean LCD codes were first studied in \cite{Massey1992} where the existence of asymptotically good LCD codes was established. In their 2014 paper \cite{Bringer2014}, Bringer {\it et al.} utilized binary Euclidean LCD codes to protect against side-channel and fault injection attacks. This practical application triggered a still growing body of literature. The works in \cite{Chen2018,Li2017,Jin2016,Jin2018,Sok2018,Liu2021,Shi2021,Carlet2018,Wu2021} and the references therein form a partial list in an already vast literature. A breakthrough on the studies of LCD codes came in the work of Carlet {\it et al.} in \cite{Carlet2018}. The main insight is that any $q$ ary linear code is equivalent to a Euclidean LCD code when $q>3$ and any $q^2$ ary linear code is equivalent to a Hermitian LCD code when $q>2$. The remaining cases motivate researchers to construct binary or ternary Euclidean LCD codes and quaternary Hermitian LCD codes. The follow-up works in \cite{Zhou2019,Bouyuklieva2021,Harada2021,Lu2020,Araya2020} exemplify the challenges that remain to be addressed and the progress made.

Linear codes with certain duality properties are important ingredients in the construction of quantum codes. Three types of quantum codes, namely, \emph{quantum (stabilizer) error-correcting codes} (QECCs) \cite{Calderbank1998,Lisonek2014}, \emph{entanglement-assisted quantum error-correcting codes} (EAQECCs) \cite{Brun2006}, and \emph{quantum subsystem codes} (QSCs) \cite{Aly2006}, also known as \emph{operator QECCs}, have parameters that can be determined from the properties of the hulls of the relevant linear codes. It is known that any Euclidean or Hermitian self-orthogonal code corresponds to a QECC. Linear codes with large-dimensional hulls can generate QECCs with good minimum distances \cite{Lisonek2014,Ezerman2019}. In a quantum communication system, the encoder and the decoder may be able to share \emph{error-free entangled bits}, often abbreviated as \emph{ebits}, prior to information exchanges. This leads to a generalization of error control in the quantum setup from QECCs to EAQECCs. Any linear code corresponds to two EAQECCs. Guenda, Jitman, and Gulliver showed in \cite{Guenda2017} that the number of pre-shared pairs of an EAQECC is determined by the dimension of the hull of the corresponding linear code.

In the classical setup, the \emph{Singleton bound} is a measure of optimality. Codes whose parameters meet the equality in the bound are labeled \emph{maximum distance separable} (MDS) codes. The bound has natural analogues in all three types of quantum error-correcting codes mentioned above. To construct MDS EAQECCs, some classes of \emph{generalized Reed-Solomon} (GRS) codes with Euclidean or Hermitian hulls of flexible dimensions have been built in \cite{Luo2019,Fang2020,Gao2021,Li2019,Tian2020}. By calculating the dimensions of the Hermitian hulls of cyclic and constacyclic classical MDS codes, several classes of MDS EAQECCs have been devised in \cite{Chen2021,Chen2021a,Qian2019,Qian2017,Wang2019}. Gan {\it et al.} in \cite{Gan2021} went further by characterizing cyclic codes with Hermitian hulls of given dimensions. Euclidean and Hermitian hulls of algebraic geometry codes have also been studied in \cite{Pereira2021,Sok2022a,Sok2022b}. Several classes of linear codes with one-dimensional Euclidean hulls were built in \cite{Qian2021,Li2019a} by using characters over abelian groups. Very recently, Luo {\it et al.} proved in \cite{Luo2021} that every $q$ ary linear code with $\ell$-dimensional Euclidean hull is equivalent to a linear code with $j$-dimensional Euclidean hull when $q>3$, for each $j\leq\ell$. With a slight adjustment, the same proof confirms that every $q^2$ ary linear code with $\ell$-dimensional Hermitian hull is equivalent to a linear code with $j$-dimensional Hermitian hull when $q>2$, for each $j\leq\ell$.

\subsection{Hermitian Hulls which are MDS}

In addition to the computation of the dimensions of the hulls of linear codes, determining their minimum distances have also been explored in \cite{SALAHA.2009,Qian2013,Du2020}. Each linear code $\cc$ leads to a quantum subsystem code $\mathcal{Q}$. Its parameters are determined by the dimension of $\cc$ and the parameters of either ${\rm Hull}_{\rm E} (\cc)^{\perp_{\rm E}}$ or ${\rm Hull}_{\rm H} (\cc)^{\perp_{\rm H}}$, as appropriate. A linear code whose Euclidean or Hermitian hull is MDS gives rise to an optimal subsystem code. Up to now, far too little attention has been paid to such linear codes. The only known construction of $q^2$ ary linear codes that are not Hermitian self-orthogonal, but whose Hermitian hulls are MDS, was studied in \cite{Qian2013} based on cyclic codes.

The Hermitian hull of an MDS code is not necessarily MDS. Constructions of GRS codes whose Hermitian hulls vary in terms of dimensions were proposed in \cite{Fang2020,Li2023}. The Hermitian hulls are of the form
\[
\left\{(a_1f(b_1),\ldots,a_nf(b_n)) :
f(x) = h(x) \prod_{i=1}^s(x-b_i) \in \Ff_q[x] \mbox{ with } \deg(h(x)) < k-s \right\},
\]
where $a_i\in\Ff_{q^2}^*$ and $b_i\in\Ff_{q^2}$ for each $1\leq i \leq n$. One can then verify that all codewords of the Hermitian hulls have at least one zero coordinate. This clearly implies that the Hermitian hulls are not MDS. The GRS codes in \cite{Fang2020,Li2023} are monomially-equivalent to Hermitian self-orthogonal codes. Thus, the constructions in \cite{Fang2020,Li2023} yield MDS QECCs as well. By the recent propagation rule in \cite[Theorem 12]{Luo2021}, the MDS EAQECCs constructed in \cite{Fang2020,Li2023} can be derived from the resulting MDS QECCs. Their MDS EAQECCs have the same minimum distances as the corresponding MDS QECCs for some fixed length. Starting with a given QECC, the said propagation rule, which we reproduce below as Lemma \ref{lem:more} for convenience, shows that the dimensions of EAQECCs obtained from QECC can be increased by $i$ with the help of $i$ pre-shared entangled pairs. In addition, the MDS EAQECCs presented in \cite{Fang2020,Li2023} improve on the dimensions of the corresponding MDS QECCs.

In general, it is interesting to construct quantum MDS codes with large minimum distances without the help of any pre-shared entangled pair, \textit{i.e.}, for MDS QECCs. With the help of entanglement, as can be seen in the table of EAQECCs in \cite{Grassl:EAQECCtables}, given pre-shared entangled pairs, EAQECCs have larger minimum distances than the comparable QECCs for fixed length and dimension.

Comparing MDS EAQECCs and MDS QECCs in a meaningful way, given the same length and dimension, is not straightforward. The gain in the minimum distance in the entanglement-assisted setup comes with a price of creating, sharing, and maintaining the purity of the pre-shared entangled pairs. An MDS code whose Hermitian hull is also MDS simultaneously leads to both an MDS QECC and an MDS EAQECC since the Hermitian hull is a Hermitian self-orthogonal code. In this paper, we point out that MDS codes whose Hermitian hulls are also MDS enable us to compare the constructed QECCs and EAQECCs directly, in terms of their trade off regarding the gain in the minimum distance and the required number of entangled pairs. Using MDS codes whose Hermitian hulls are MDS, the minimum distances of the constructed MDS EAQECCs can be improved by $s$ with the help of $2s$ pre-shared pairs with respect to the constructed MDS QECCs of the same length and dimension.

\subsection{Our Contributions}

In this paper, we investigate linear codes whose Hermitian hulls are MDS. The techniques and results can be summarized as follows.
\begin{enumerate}[1.]
\item Focusing on the Hermitian hulls of linear codes, we develop tools to study linear codes whose Hermitian hulls are MDS. We construct new families of such linear codes and highlight the novelty by comparing the codes to those covered in prior works.

\begin{itemize}
	\item Lemma \ref{lemma8} presents necessary and sufficient conditions for an $[m,k]_{q^2}$ GRS code, for $m \le q^2$, to have a Hermitian hull that contains an $[m,\ell]_{q^2}$ GRS code, for $\ell \le k$. The key idea is to examine the punctured code $\mathcal{P}(\cc)$, as defined by Rains in \cite{Rains1999}, of a given $[q^2,k]_{q^2}$ code $\cc$. We can apply the puncturing technique whenever a suitable codeword of Hamming weight $m$ in $P(\cc)$ can be identified.
	
	\item Theorem \ref{GRSMDSHULL} simplifies the sufficient conditions of Lemma \ref{lemma8} and provides a sufficient condition for an $[m,k]_{q^2}$ GRS code, for $m \le q^2$, to have as its Hermitian hull an $[m,k-1]_{q^2}$ GRS code. We use the theorem to build new classes of Hermitian hulls in Theorems \ref{GRScon1} to \ref{GRScon4}. These classes contain $(k-1)$-dimensional MDS codes of shorter lengths which are derived from $k$-dimensional MDS codes of length $q^2$.
\end{itemize}

\item The Euclidean dual of an algebraic geometry (AG) code is well-known, {\it e.g.}, in the treatment by Stichtenoch in \cite{Stich88}. The Hermitian dual of a general $q^2$ ary AG code has been less understood. Prior works, especially in the context of quantum stabilizer codes, tend to focus on Hermitian self-orthogonality. A recent work of Pereira {\it et al.} on the Hermitian hull of AG codes in \cite{Pereira2021} has not focused on the MDS case. The hull dimensions have not been explicitly determined.

We generalize further to accomplish the following tasks.
\begin{itemize}

\item We explicitly determine the Hermitian hull dimension of a special class of two-point rational AG codes in Theorem \ref{thm:rankcomp}. To obtain such a result, we carefully select a set of evaluation points such that the residues of the Weil differential form on the points are elements in $\Ff_q$. The main tool to derive the dimension of the Hermitian hull is the celebrated Riemann-Roch Theorem. Theorem \ref{thm:arbitrary-hull} extends the result of Theorem \ref{thm:rankcomp} to get a Hermitian hull of an \emph{arbitrary} dimension.

\item Corollaries \ref{cor:1} to \ref{cor:3} provide more explicit constructions of families of linear codes whose Hermitian hulls are MDS. Most known results on the Hermitian hulls stem from Reed-Solomon-type codes, which are one-point rational AG codes. We obtain our results from two-point rational AG codes. Our codes are non Reed-Solomon-type and, on the same length and dimension, are not monomially equivalent to codes of the Reed-Solomon-type. In particular, the codes we have constructed are not monomially equivalent to known Hermitian self-orthogonal codes. We note that a more general equivalent transformation may exist based on the respective representations as AG codes, the required computation to verify equivalence in general is significantly more complex.
	
Based on Hermitian self-orthogonal codes, Theorem \ref{thm:MDSHull} paves the way to enlarge the dimensions of the Hermitian hulls of some known GRS codes in the literature. Such an improvement leads to EAQECCs with new parameters that, when we fix the length and the dimension, can be shown to improve the error-control capability of known stabilizer codes, with some trade-off.
\end{itemize}
\end{enumerate}

This paper is organized as follows. We collect notation and definitions of cyclic codes and related AG codes as well as useful known results in the preliminary Section \ref{sec:pre}. Section \ref{sec:3} contains the main results of the paper; Subsection \ref{subsec:GRS typed} deals with Hermitian hulls of linear codes from related cyclic codes and GRS codes; Subsection \ref{subsec:non-GRS typed} applies two-point rational AG codes to construct linear codes whose Hermitian hulls are MDS and have arbitrary dimensions. Section \ref{sec:4} discusses application in quantum error-correcting codes. The last section concludes the paper and proposes some possible research directions.

\section{Preliminaries}\label{sec:pre}

This section recalls some basic definitions and results about cyclic codes, GRS codes, and algebraic geometry codes. Given two integers $m$ and $n$, let $[m,n]$ denote the set $\{m,m+1,\ldots,n\}$, if $m\leq n$, and the empty set $\emptyset$, if $m>n$.

\subsection{Cyclic Codes and Generalized Reed-Solomon codes}\label{subsec:cyclicGRS}

An $[n,k,d]_q$ code $\cc$ is {\it cyclic} if $(c_{n-1},c_0,\ldots,c_{n-2})\in\cc$ for each $(c_0,c_1,\ldots,c_{n-1})\in\cc$. Each codeword $(c_0,c_1,\ldots,c_{n-1})$ of $\cc$ can be represented by a polynomial $c(x)=c_0+c_1x+\ldots+c_{n-1}x^{n-1} \in \Ff_q[x]$. This representation allows for the identification of $\cc$ as an ideal in $\Ff_q[x]/\langle x^n-1\rangle$. The monic polynomial $g(x)$ of degree $n-k$ in the ideal is a divisor of $x^n-1$ and is called the \emph{generator polynomial} of $\cc$.

We denote by ${\rm ord}_n(q)$ the smallest positive integer $r$ such that $q^r \equiv 1 \pmod{n}$. If $\alpha$ is a primitive $n^{\rm th}$ root of unity in $\Ff_{q^r}$, then the roots of $x^n-1$ are $\alpha^j$ for $j \in[0,n-1]$. The collection $D=\{i\in[0,n-1]:g(\alpha^i)=0\}$ is the \emph{defining set} of a cyclic code $\cc$. Its complement $[0,n-1]\setminus D$ is the \emph{generating set} of $\cc$. Since $D$ is the union of some $q$ cyclotomic cosets modulo $n$, we can assume that the number of such $q$ cyclotomic cosets is $t$ and write $D=\bigcup_{\ell=1}^tC_{i_\ell}$ with $C_{i_\ell}$ being $q$ cyclotomic coset containing $i_\ell$. Using the $t\times n$ matrix
\[
H=\begin{pmatrix}
1 & \alpha^{i_1} &\cdots & \alpha^{i_1(n-1)}\\
1 & \alpha^{i_2} &\cdots & \alpha^{i_2(n-1)}\\
\vdots& \vdots &  \ddots & \vdots\\
1 & \alpha^{i_{t}} &\cdots & \alpha^{i_{t}(n-1)}\\
\end{pmatrix},
\]
a codeword $\mathbf{c}$ is in the cyclic code $\cc$ if and only if $H\mathbf{c}^{\top}=\mathbf{0}$. Let $\Tr_m$ denote the \emph{trace function} from $\Ff_{q^m}$ onto $\Ff_q$. The cyclic code $\cc$ has the following representation in terms of the trace function.
\begin{lem}{\rm (\cite{Delsarte1975})}\label{trace}
Let $\cc$ be a cyclic code of length $n$ over $\Ff_q$. Let $r={\rm ord}_n(q)$ and let $\alpha$ be a primitive $n^{\rm th}$ root of unity in $\Ff_{q^r}$. Let the generating set of $\cc$ be $\bigcup_{\ell=1}^sC_{i_\ell}$, where $C_{i_\ell}$ is the $q$ cyclotomic cosets modulo $n$ that contains $i_\ell$. If $|C_{i_\ell}|= m_\ell$ for any $\ell\in[s]$, then
\[
\cc=\left\{\left(\sum_{\ell=1}^s \Tr_{m_\ell}(\theta_\ell\alpha^{-ui_\ell})\right)_{u\in[0,n-1]}: \theta_\ell\in\Ff_{q^{m_\ell}},1\leq\ell\leq s\right\}.
\]
\end{lem}

The Hartmann-Tzeng bound in \cite{Hartmann1972} is a \emph{lower bound} on the minimum distance of a cyclic code.

\begin{lem}{\rm (Hartmann-Tzeng bound)}\label{HTbound} Let $\cc$ be an $[n,k,d]_q$ cyclic code with defining set $D$. Let $a$, $b$, and $c$ be integers such that $\gcd(b,n)=1$ and $\gcd(c,n)=1$. If $\{a+bi_1+ci_2 : i_1 \in[0,x-2] \mbox{ and } i_2\in[0,y]\}$ is contained in $D$, then $d\geq x+y$.
\end{lem}

The \emph{extended code} of the cyclic code $\cc$ is
\[
\mathcal{E}(\cc)=\left\{(c_0,\ldots,c_{n-1},c_n):(c_0,\ldots,c_{n-1})\in\cc \mbox{ and } \sum_{i=0}^n c_i =0\right\}.
\]
From \cite[Subsection 1.5.2]{huffman2003}, the code $\mathcal{E}(\cc)$ has length $n+1$, dimension $k$, and a parity-check matrix
\begin{equation}\label{extendedcyclic}
H_{\mathcal{E}(\cc)} =
\begin{pmatrix}
1 & 1 &\cdots & 1 & 1\\
1 & \alpha^{i_1} &\cdots & \alpha^{i_1(n-1)}& 0\\
\vdots& \vdots &  \ddots & \vdots&\vdots\\
1 & \alpha^{i_{t}} &\cdots & \alpha^{i_{t}(n-1)}& 0\\
\end{pmatrix}.
\end{equation}

Reed-Solomon (RS) codes form the most famous class of cyclic codes. Using tools from finite fields, a generalization of RS codes emerges naturally. We denote the multiplicative group of $\Ff_q$ by $\Ff_q^*$. Let $b_1,\ldots,b_n$ be distinct elements of $\Ff_q$ that define a vector $\mathbf{b}=(b_1,\ldots,b_n)$ of length $n$. Let $\mathbf{a}=(a_1,\ldots,a_n) \in  (\Ff_q^*)^n$. For each $k\in[0,n]$, the generalized Reed-Solomon (GRS) code $GRS_k(\mathbf{b},\mathbf{a})$ is
\begin{equation}\label{eq:GRS}
	GRS_k(\mathbf{b},\mathbf{a}) := \left\{(a_1f(b_1),\ldots,a_nf(b_n)) : f(x) \in\Ff_q[x] \mbox{, with } \deg(f(x)) < k \right\}.
\end{equation}
From \cite[Chapter 9]{Ling2004}, we know that the code $GRS_k(\mathbf{b},\mathbf{a})$ is an $[n,k,n-k+1]_q$ MDS code with a generator matrix
\begin{equation}\label{GRSgenerator}
G_k=\begin{pmatrix}
	a_1 & a_2 &\cdots & a_n\\
	a_1b_1 & a_2b_2 &\cdots & a_nb_n\\
	\vdots& \vdots &  \ddots & \vdots\\
	a_1b_1^{k-1} & a_2b_2^{k-1} &\cdots & a_nb_n^{k-1}\\
\end{pmatrix}.
\end{equation}
As explained in the same reference, the dual code of $GRS_k(\mathbf{b},\mathbf{a})$ is also a GRS code.

\subsection{Algebraic Geometry Codes}

We recall useful notions and facts on AG codes. For brevity, we refer the readers to Stichtenoth's textbook \cite{Stich} for terms related to algebraic function fields that we have left undefined here.

Let $\cal{X}$ be a \emph{smooth projective curve}  over $\F_q$. Let $\F_q(\cal{X})$ denote the field of \emph{rational functions} of ${\cal X}$. Function fields of algebraic curves over $\F_q$ are
finite separable extensions of $\F_q(x)$, where $x$ is transcendental over $\F_q$. We identify points on ${\cal X}$ with \emph{places} of the function field $\F_q(\cal{X})$. A point $P$ on $\cal{X}$ is \emph{rational} if all of
its coordinates belong to $\F_q$. Rational points are places of degree one. The set of $\F_q$ rational points on $\cal{X}$ is denoted by  $\mathcal{X} (\F_q)$. For an element $\alpha \in \F_q$, we denote by $P_{\alpha}$ and $O$, respectively, the \emph{zero place} of $x-\alpha$ and the \emph{pole place} of $1/x$ .

A \emph{divisor} $G$ on $\cal{X}$ is the formal sum
$\sum_{P\in \cal{X}} n_P P$ with only finitely many nonzeros $n_P \in \mathbb{Z}$
and $G$ is \emph{rational} if $G^\sigma = G$ for any permutation $\sigma$ in the Galois group $\Gal(\overline{\F}_q/\F_q)$.
The \emph{support} of $G$ is $\supp(G):= \{ P \in {\cal{X}} \, : \, n_P \neq 0\}$. For $G=\sum_{P\in \cal{X}} n_P P$, the \emph{degree} of $G$ is $\deg(G):=\sum_{P\in \cal{X}} n_P \deg(P)$, with $\deg (P)$ being the cardinality of the orbit of $P$ under the action of $\sigma$.
Given two divisors $G=\sum_{P\in \cal{X}} n_P P$ and
$H =\sum_{P\in \cal{X}} m_P P$, we can define the operators $\wedge$ and $\vee$ on the divisors as follow
\[
G \wedge H:=\sum_{P\in \cal{X}} \min (m_P,n_P) \, P \mbox{ and }
G \vee H:= \sum_{P\in \cal{X}} \max (m_P,n_P) \, P.
\]
We say that $G \ge H$ if $n_P \ge m_P$ for \emph{all} places $P \in \cal{X}$.

Any nonzero rational function $f$ on the curve $\cal X$ can be expressed uniquely in terms of its irreducible factors as
\[
f(x) = \alpha \prod_{i=1}^s p_i(x)^{e_i} \mbox{, with } \alpha \in \F_q^{*}.
\]
Each irreducible polynomial $p_i(x)$ corresponds to a place called $P_i$. The \emph{valuation} of $f$ at $P_i$ is $v_{P_i}(f) := t$ if $p_i(x)^t$ divides $f(x)$ but $p_i(x)^{t+1}$ does not divide $f(x)$. The \emph{principal divisor} of $f$ is $(f):= \sum_{P\in \cal{X}} v_P(f) P$, where $v_P (f)$ is the \emph{normalized discrete valuation} corresponding to the place $P$. Let $Z(f)$ and $N(f)$ denote, respectively, the set of zeros and poles of $f$. The \emph{zero divisor} and \emph{pole divisor} of $f$ are given, respectively, by
\[
(f)_0:=\sum_{P\in Z(f)} v_{P}(f) \, P \mbox{ and }
(f)_\infty:=\sum_{P\in N(f)}-v_{P} (f) \, P.
\]
We can then write $(f)=(f)_0 - (f)_\infty$. The \emph{degree} of any principal divisor $(f)$ is $0$. Divisors $G$ and $H$ are \emph{equivalent} if there exists a rational function $f$ such that $G-H=(f)$.

Let $\Omega:=\{f \, dx \, : \, f\in \F_q(\cal{X})\}$ be the set of \emph{differential forms} on $\cal X$. For an $\alpha\in \F_q$ and $f \in \F_q(\cal{X})$ with $v_{P_\alpha}(f) \geq -1$, we can expand $f(x)$ in the neighborhood of $\alpha$ as
\[
f(x) = \ldots + \frac{a_{-1}}{x-\alpha} + a_0 + a_1(x-\alpha)+\ldots.
\]
Thus, at $P_\alpha$, the \emph{residue} ${\rm Res}_{P_\alpha}(f \,dx)$ is $a_{-1}$. We know from \cite[Proposition 1.2.1]{Stich} that, over $\F_{q^2}(x)$, the places consist of $q^2+1$ places of degree one, the pole place $O$, and the $q^2$ zero places corresponding to the elements of $\F_{q^2}$.

For a divisor $G$ on $\cal X$, we define
\[
{\cal L}(G):=\{f\in \F_q({\cal{X}}) \setminus \{0\} \, : \,
(f) + G \ge 0\} \cup \{0\} \mbox{ and }
{\Omega}(G):=\{\omega\in {\Omega \setminus \{0\}} \, : \,
({\omega})- G \ge 0\} \cup \{0\}.
\]
Both ${\cal L}(G)$ and ${\Omega}(G)$ are finite dimensional vector spaces over $\mathbb{F}_q$. For a differential form $\omega$ on $\cal X$, there exists a unique rational function $f$ on $\cal X$ such that $\omega = f \, dt$, with $t$ being a \emph{local uniformizing parameter}. The divisor class of a nonzero differential form is called the \emph{canonical divisor}. If $W$ is the canonical divisor and $\cal{X}$ is a curve of \emph{genus} $g$, then $\deg(W) = 2 \, g-2$.

The dimension of ${\cal L}(G)$ is determined by the Riemann-Roch theorem.

\begin{prop}\textnormal{\cite[Theorem 1.5.15 (Riemann-Roch)]{Stich}}\label{thm:Riemann-Roch} Let $g$ be the genus of a given smooth algebraic curve $\cal{X}$. If $W$ is a canonical divisor, then, for each divisor $G$,
\begin{equation}\label{thm:RR}
\dim( {\cal L}(G)) = \deg(G) + 1 -g + \dim ({\cal L}(W-G)).
\end{equation}
\end{prop}

Let $D:=P_{1}+\ldots+P_{n}$, with $P_{i}$ being a place of degree $1$ for each $i \in [1,n]$. Let $G$ be a divisor such that $\supp(D) \cap \supp(G) = \emptyset$. The \emph{evaluation AG code} on divisors $D$ and $G$ is
\[
\cc_{\cal L}(D,G):=\{(f(P_{1}),\ldots,f(P_{n})) \, : \,
f \in {\cal L}(G)\}.
\]
Let $\Res_{P}(\omega)$ denote the \emph{residue} of $\omega$ at point $P$. The \emph{differential AG code} on divisors $D$ and $G$ is
\[
\cc_{\Omega}(D,G):= \{(\Res_{P_{1}}(\omega),\ldots,\Res_{P_{n}}(\omega)) \, : \, \omega \in {\Omega}(G-D)\}.
\]

The parameters of $\cc_{\cal L}(D,G)$ can be easily computed.
\begin{prop}\textnormal{\cite[Theorem 2.2.2, Corollary 2.2.3]{Stich}}\label{thm:distance}
The dimension $k$ and minimum distance $d$ of $\cc_{\cal L}(D,G)$ are
\begin{equation}
k= \dim({\cal L}(G)) + \dim({\cal L}(G-D)) \mbox{ and } d \ge n-\deg (G).
\end{equation}
\item If $2g -2 < \deg(G) < n$, then $\cc_{\cal L}(D,G)$ has
\begin{equation}
k=\deg (G)-g+1\textnormal{ and } d\ge n-\deg (G).
\label{eq:distance}
\end{equation}
\end{prop}

An evaluation code and its corresponding differential code are closely connected by the Euclidean inner product.

\begin{lem}\textnormal{\cite[Theorem 2.2.8]{Stich}}\label{lem:dual1} $\cc_{\cal L}(D,G)^{\perp_{\rm E}} = \cc_{\Omega}(D,G)$.
\end{lem}

The differential code $\cc_{\Omega}(D,G)$ can be expressed as an evaluation code.
\begin{lem}\textnormal{\cite[Proposition 2.2.10] {Stich}}\label{lem:dual2} If $\omega $ is a differential form with the property that, for each $i \in [1,n]$,
\[
v_{P_i}(\omega)=-1 \mbox{ and } \Res_{P_i}(\omega) \neq 0,
\]
then there exists an ${\bf a}=(a_1,\ldots,a_n) \in ({\F^*_q}){^n}$ such that
\[
\cc_{\Omega}(D,G) = {\bf a} \cdot \cc_{\cal L}(D, D-G + (\omega)),
\]
where ${\bf a} \cdot \cc:=\{(a_1 c_1,\ldots,a_n c_n) : (c_1,\ldots,c_n) \in \cc\}$.
\end{lem}

Any GRS code $GRS_k(\mathbf{b},\mathbf{a})$ as defined in \eqref{eq:GRS} is monomially-equivalent to a rational AG code $\cc_{\cal L}(D,G)$ with $\deg(G)=k-1$ as stated in \cite[Proposition 2.3.3]{Stich}. The parameters are related as follows. For any $1 \leq i \leq n$, let $b_i :=x(P_i)$ and $a_i :=u(P_i)$ for some $u(x) \in \Ff_{q^2}(x)$ that satisfies $(u)=(k-1) O-G$. For $0 \leq j \leq k-1$, the vectors
\[
(u x^j(P_1),\ldots,u x^j(P_{n})) = (a_1b_1^j,\ldots,a_{n}b_{n}^j)
\]
constitute a basis of $\cc_{\cal L}(D,G)$ and, hence, a generator matrix of $\cc_{\cal L}(D,G)$ is in fact the matrix $G_k$ in \eqref{GRSgenerator}.

\section{Main results}\label{sec:3}

This section investigates the Hermitian hulls of GRS codes and AG codes. Several constructions of linear codes whose Hermitian hulls are MDS are then proposed.

\subsection{MDS Hermitian Hulls from GRS Codes}\label{subsec:GRS typed}
This subsection is devoted to the constructions of linear codes whose Hermitian hulls are MDS, based on GRS codes and their punctured codes.
Rains in \cite{Rains1999} defined the code $\mathcal{P}(\cc)$ based on a given $[n,k,d]_{q^2}$ code $\cc$, to serve as a powerful ingredient in a construction of Hermitian self-orthogonal codes, as
\begin{equation}\label{eq:RainP}
\mathcal{P}(\cc) = \left\{(a_1,\ldots,a_n) \in \Ff_q^n \, : \, \sum_{i=1}^n a_i \, u_i \, v_i^q=0 \mbox{ for all } \mathbf{u}=(u_1,\ldots,u_n),\mathbf{v}=(v_1,\ldots,v_n) \in \cc \right\}.
\end{equation}
The following result was established by Ball and Vilar.
\begin{lem}{\rm \cite{Ball2022}}\label{lem31}
Given an $[n,k,d]_{q^2}$ code $\cc$, there exists a punctured code of $\cc$ of length $m \leq n$ which is monomially-equivalent to a Hermitian self-orthogonal code if and only if there exists a codeword $\mathbf{c}$ of weight $m$ in $\mathcal{P}(\cc)$.
\end{lem}

Lemma \ref{lem31} motivates us to give a similar characterization of a GRS code whose Hermitian hull contains another GRS code. Let $\alpha$ be a primitive element of $\Ff_{q^2}$. Choosing $\mathbf{1}=(1,\ldots,1)$ and $\mathbf{b} = (\alpha^0,\ldots,\alpha^{q^2-2},0)$, we build
$GRS_k(\mathbf{b},\mathbf{1})$ over $\Ff_{q^2}$. This is a $[q^2,k,q^2-k+1]_{q^2}$ code whose Hermitian hull contains $GRS_\ell(\mathbf{b},\mathbf{1})$ as a subcode for $\ell \leq k$. Let
\begin{equation}\label{GRS1}
\mathcal{P}(GRS_{k,\ell}) : = \left\{(a_1,\ldots,a_{q^2})\in\Ff_{q}^n \, : \, \sum_{i=1}^{q^2} a_i \, u_i \, v_i^q=0 \mbox{ for all } \mathbf{u}\in GRS_\ell(\mathbf{b},\mathbf{1}) \mbox{ and for all } \mathbf{v}\in GRS_k(\mathbf{b},\mathbf{1})\right\}.
\end{equation}
By definition, $\mathcal{P}(GRS_{k,\ell})$ has a parity-check matrix
\[
H_{\mathcal{P}(GRS_{k,\ell})} =
\begin{pmatrix}
1 & 1 &\cdots & 1 & 1\\
1 & \alpha^{s_1} &\cdots & \alpha^{s_1(q^2-2)}& 0\\
\vdots& \vdots &  \ddots & \vdots&\vdots\\
1 & \alpha^{s_{t}} &\cdots & \alpha^{s_{t}(q^2-2)}& 0\\
\end{pmatrix},
\]
where
\[
\{s_1,\ldots,s_t\} = \{i+qj:i\in[0,\ell-1],j\in[0,k-1]\} \setminus\{0\}.
\]
If $k\leq q$, then we infer from (\ref{extendedcyclic}) that $\mathcal{P}(GRS_{k,\ell})$ is the extended code $\mathcal{E}(\mathcal{D}_{k,\ell})$ of a $q$ ary cyclic code $\mathcal{D}_{k,\ell}$ of length $q^2-1$ with defining set
\begin{equation}\label{eq:defset}
\{i+qj: i\in[0,\ell-1],j\in[\ell,k-1]\} \cup \{i+qj:i\in[\ell,k-1],j\in[0,\ell-1]\} \cup \left(\{i+qj:i\in[0,\ell-1],j\in[0,\ell-1]\} \setminus\{0\}\right).
\end{equation}
If $\ell=k$, then the defining set in \eqref{eq:defset} becomes $\{i+qj:i\in[0,\ell-1],j\in[0,\ell-1]\} \setminus\{0\}$ since $[k,k-1]=\emptyset$. By Lemma \ref{HTbound}, $\mathcal{E}(\mathcal{D}_{k,\ell})$ is a
$[q^2,q^2-2\ell k+\ell^2,d\geq k+\ell-1]_q$ code.

\begin{lem}\label{lemma8}
Let $\alpha$ be a primitive element of $\Ff_{q^2}$ and let $\mathbf{b}=(\alpha^0,\ldots,\alpha^{q^2-2},0)$. Let $\mathcal{P}(GRS_{k,\ell})$ in (\ref{GRS1}) be the extended code $\mathcal{E}(\mathcal{D}_{k,\ell})$ of $\mathcal{D}_{k,\ell}$ whose defining set is as in (\ref{eq:defset}). Let $\ell \leq k\leq q$. Let $\widehat{\mathbf{a}}$ and $\widehat{\mathbf{b}}$ be vectors in $\Ff_{q^2}^m$ with $m \leq q^2$. There exists a $q^2$ ary GRS code $GRS_k(\widehat{\mathbf{b}},\widehat{\mathbf{a}})$ of length $m$ whose Hermitian hull contains $GRS_\ell(\widehat{\mathbf{b}},\widehat{\mathbf{a}})$ as a subcode if and only if the following three conditions are met.
\begin{enumerate}
    \item There exists a codeword $\mathbf{x}:=(x_1,\ldots,x_{q^2})$ of weight $m$ with nonzero entries $x_{i_1},\ldots,x_{i_m}$ in $\mathcal{E}(\mathcal{D}_{k,\ell})$.
    \item The vector $\widehat{\mathbf{a}}$ is of the form $(a_{i_1},\ldots,a_{i_m})$ with $a_{i_j}^{q+1} = x_{i_j}$ for each $j\in[1,m]$.
\item The vector $\widehat{\mathbf{b}}$ is obtained by deleting the $q^2-m$ entries in $\mathbf{b}$ whose coordinates are indexed by $[1,q^2]\setminus\{i_1,\ldots,i_m\}$.
\end{enumerate}
\end{lem}
\begin{proof}
Let there be a codeword $(x_1,\ldots,x_n)$ of weight $m$ with nonzero entries $x_{i_1},\ldots,x_{i_m}$ in $\mathcal{E}(\mathcal{D}_{k,\ell})$. Since $x_{i_j}\in\Ff_q^*$, there exists a nonzero element $a_{i_j} \in\Ff_{q^2}^*$ such that $x_{i_j}=a_{i_j}^{q+1}$ for each $j\in[1,m]$. We note that $m\geq k+\ell-1\geq k$. Using the vector $\mathbf{1}=(1,\ldots,1)$ of length $q^2$, we obtain the respective codes $GRS_k(\widehat{\mathbf{b}},\widehat{\mathbf{1}})$ and $GRS_\ell(\widehat{\mathbf{b}},\widehat{\mathbf{1}})$ of length $m$ by puncturing $GRS_k(\mathbf{b},\mathbf{1})$ and $GRS_\ell(\mathbf{b},\mathbf{1})$ on $S=[1,q^2]\setminus\{i_1,\ldots,i_m\}$. By (\ref{GRS1}), for any $(u_{i_1},\ldots,u_{i_m})\in GRS_\ell(\widehat{\mathbf{b}},\widehat{\mathbf{1}})$ and any $(v_{i_1},\ldots,v_{i_m})\in GRS_k(\widehat{\mathbf{b}},\widehat{\mathbf{1}})$, we have
\[
\sum_{j=1}^mx_{i_j}u_{i_j}v_{i_j}^q = \sum_{j=1}^ma_{i_j}u_{i_j}(a_{i_j}v_{i_j})^q=0.
\]
Writing $\widehat{\mathbf{a}}=(a_{i_1},\ldots,a_{i_m})$, we easily confirm that $GRS_\ell(\widehat{\mathbf{b}},\widehat{\mathbf{a}}) \subseteq GRS_k(\widehat{\mathbf{b}},\widehat{\mathbf{a}})^{\perp_{\rm H}}$, which implies that $GRS_\ell(\widehat{\mathbf{b}},\widehat{\mathbf{a}})$ is a subcode of the Hermitian hull of $GRS_k(\widehat{\mathbf{b}},\widehat{\mathbf{a}})$.

Conversely, the desired conclusion follows by reversing the direction of the argument above.
\end{proof}

Lemma \ref{lemma8} directly leads to the following technique to generate GRS codes whose Hermitian hulls are MDS.

\begin{thm}\label{GRSMDSHULL}
Let the notation be as in Lemma \ref{lemma8}. If there exists a codeword $\mathbf{x}=(x_1,\ldots,x_{q^2})$ of weight $m$ with nonzero entries $x_{i_1},\ldots,x_{i_m}$ in $\mathcal{E}(\mathcal{D}_{k,k-1}) \setminus \mathcal{E}(\mathcal{D}_{k,k})$, then the Hermitian hull of $GRS_k(\widehat{\mathbf{b}},\widehat{\mathbf{a}})$ is $GRS_{k-1}(\widehat{\mathbf{b}},\widehat{\mathbf{a}})$.
\end{thm}
\begin{proof}
By Lemma \ref{lemma8}, the Hermitian hull of $GRS_k(\widehat{\mathbf{b}},\widehat{\mathbf{a}})$ contains a subcode $GRS_{k-1}(\widehat{\mathbf{b}},\widehat{\mathbf{a}})$. Since the codeword $\mathbf{x}$ does not belong to $\mathcal{D}_{k,k}$, it follows from the lemma that $GRS_{k}(\widehat{\mathbf{b}},\widehat{\mathbf{a}})$ is not a subcode of ${\rm Hull}_{\rm H} (GRS_k(\widehat{\mathbf{b}},\widehat{\mathbf{a}}))$. This implies that $GRS_k(\widehat{\mathbf{b}},\widehat{\mathbf{a}})$ is not Hermitian self-orthogonal and that the dimension of ${\rm Hull}_{\rm H}(GRS_k(\widehat{\mathbf{b}},\widehat{\mathbf{a}}))$ is $\leq k=1$. Thus, the Hermitian hull of $GRS_k(\widehat{\mathbf{b}},\widehat{\mathbf{a}})$ is $GRS_{k-1}(\widehat{\mathbf{b}},\widehat{\mathbf{a}})$.
\end{proof}

\begin{remark}\label{rem:Chen}
Chen in \cite{Chen2023} raised a question on the maximal Hermitian hull dimension of linear codes under monomial equivalence. In other words, given a linear code, find a monomially-equivalent code whose Hermitian hull has the largest dimension among all monomially-equivalent codes. It is straightforward to verify that the GRS codes constructed by Theorem \ref{GRSMDSHULL} has $(k-1)$-dimensional Hermitian hull. The codes are not Hermitian self-orthogonal under any monomial equivalence. Thus, the \emph{maximal} Hermitian hull dimension of any GRS code that can be obtained from Theorem \ref{GRSMDSHULL} is $k-1$.
\end{remark}

\begin{remark}\label{remark1}
The respective dimensions of $\mathcal{E}(\mathcal{D}_{k,k-1})$ and $\mathcal{E}(\mathcal{D}_{k,k})$ are $q^2-k^2+1$ and $q^2-k^2$. Hence, there exists a nonzero codeword in $\mathcal{E}(\mathcal{D}_{k,k-1})\setminus \mathcal{E}(\mathcal{D}_{k,k})$ whenever $k\leq q$. If $k<q$, then Lemma \ref{trace} allows us to express $\mathcal{E}(\mathcal{D}_{k,k-1})\setminus \mathcal{E}(\mathcal{D}_{k,k})$ as
\begin{equation}\label{eqtr}
\left\{\left(c_0,\ldots,c_{q^2-2},\sum_{r=0}^{q^2-2}c_r\right):\theta_{k-1,k-1}\in\Ff_{q}^*,\theta_{t,t}\in\Ff_{q},t\in[k,q-1], \mbox{ and } \theta_{i,j}\in\Ff_{q^2} \mbox{ for } (i,j) \in T \right\},
\end{equation}
with
\[
T=\bigcup_{i=k}^{q-1}\left\{(i,j):j\in[0,k-1]
\bigcup[i+1,q-1]\right\},
\]
since $[q,q-1]$ is the empty set, and
\[
c_r = \sum_{t=k-1}^{q-1}\theta_{t,t}\alpha^{-r t(q+1)} + \sum_{(i,j)\in T} \Tr_2\left(\theta_{i,j}\alpha^{-r(i+qj))}\right) \mbox{ for each } r\in[0,q^2-2].
\]
We can then confirm that $\sum_{r=0}^{q^2-2}c_r = -\theta_{q-1,q-1}$.
\end{remark}

Setting $k=q$ in Theorem \ref{GRSMDSHULL} gives a construction of GRS codes whose Hermitian hulls are MDS of dimension $q-1$.

\begin{thm}\label{GRScon1}
Let $\mathbf{1}=(1,1,\ldots,1)$ and
$\mathbf{b}= (1,\alpha^1,\ldots,\alpha^{q^2-2},0)$ be vectors of length $q^2$, with $\alpha$ being a primitive element of $\Ff_{q^2}$. The Hermitian hull of $GRS_q(\mathbf{b},\mathbf{1})$ is $GRS_{q-1}(\mathbf{b},\mathbf{1})$.
\end{thm}
\begin{proof}
The respective dimensions of $\mathcal{E}(\mathcal{D}_{q,q-1})$ and $\mathcal{E}(\mathcal{D}_{q,q})$ are $1$ and $0$. By definition,
\[
\mathcal{E}(\mathcal{D}_{q,q-1})=\{(a,\ldots,a):a\in\Ff_{q}\}.
\]
By Theorem \ref{GRSMDSHULL}, ${\rm Hull}_{\rm H}(GRS_q(\mathbf{b},\mathbf{1})) = GRS_{q-1}(\mathbf{b},\mathbf{1})$.
\end{proof}

\begin{remark} The authors of \cite{Fang2020} constructed, from a $[q^2,k]_{q^2}$ GRS code with $1 \le k\le q-1$, MDS codes with arbitrary Hermitian hull dimensions $\ell$ in the range of $0 \leq \ell \leq k$. They left the case of $q \leq k \leq \lfloor \frac{n}{2}\rfloor$ as a research problem. The code $GRS_{q}(\mathbf{b},\mathbf{1})$ in Theorem \ref{GRScon1} has dimension $q$ and, hence, differ from the codes in \cite{Fang2020}.
\end{remark}

Referring to Remark \ref{remark1}, we turn the task of finding a nonzero codeword of $\mathcal{E}(\mathcal{D}_{k,k-1})\setminus \mathcal{E}(\mathcal{D}_{k,k})$ into counting the number of roots of a polynomial. As the $\theta_{i,j}$ in \eqref{eqtr} traverses the elements of $\Ff_{q^2}$, we obtain the following constructions.

\begin{thm}\label{GRScon2}
Let $k$ be an integer such that $1<k<q$. Let $\alpha$ be a primitive element of $\Ff_{q^2}$. Choosing vectors of length $q^2-1$
\[
\mathbf{a} = \left(\alpha^{-0(k-1)}, \alpha^{-1(k-1)},\alpha^{-2(k-1)},\ldots,\alpha^{-(q^2-2)(k-1)}\right)
\mbox{ and }
\mathbf{b}=\left(1, \alpha, \alpha^2,\ldots,\alpha^{q^2-2}\right),
\]
we obtain ${\rm Hull}_{\rm H}(GRS_k(\mathbf{b},\mathbf{a})) = GRS_{k-1}(\mathbf{b},\mathbf{a})$.
\end{thm}
\begin{proof}
By setting $\theta_{k-1,k-1}=1$ and $\theta_{t,t}=\theta_{i,j}=0$ for each $t\in[k,q-1]$ and each $(i,j)\in T$ in \eqref{eqtr}, we confirm that
\[
\left(\alpha^{-0(k-1)(q+1)},\alpha^{-1(k-1)(q+1)},\ldots, \alpha^{-(q^2-2)(k-1)(q+1)},0 \right) \in \mathcal{E}(\mathcal{D}_{k,k-1})\setminus \mathcal{E}(\mathcal{D}_{k,k}).
\]
By Theorem \ref{GRSMDSHULL}, ${\rm Hull}_{\rm H}(GRS_k(\mathbf{b},\mathbf{a})) = GRS_{k-1}(\mathbf{b},\mathbf{a})$.
\end{proof}

\begin{thm}\label{GRScon3}
Let $k$ be a positive integer such that $k <q$. Let $\alpha$ be a primitive element of $\Ff_{q^2}$. Let
\[
s=\gcd(k-1,q-1) \mbox{ and } B=\left\{i+\frac{q-1}{s}j \, : \, i\in \left[\frac{q-1}{s}-1\right] \mbox{ and } j\in[0,(q+1)s-1]\right\}.
\]
If $\mathbf{a}=((a_\ell)_{\ell\in B},a)$ and $\mathbf{b}=((\alpha^\ell)_{\ell\in B},0)$ are vectors of length $q^2-s(q+1)$ with $a^{q+1}=-1$ and
$a_\ell^{q+1} = \alpha^{-\ell(k-1)(q+1)}-1$ for each $\ell\in B$, then ${\rm Hull}_{\rm H}(GRS_k(\mathbf{b},\mathbf{a})) = GRS_{k-1}(\mathbf{b},\mathbf{a})$.
\end{thm}
\begin{proof}
In (\ref{eqtr}), let $\theta_{k-1,k-1}=1$, $\theta_{q-1,q-1}=-1$, and $\theta_{t,t}=\theta_{i,j}=0$ for each $t\in[k,q-2]$ and each $(i,j) \in T$. Then the corresponding codeword in $\mathcal{E}(\mathcal{D}_{k,k-1})\setminus \mathcal{E}(\mathcal{D}_{k,k})$ is
\[
\mathbf{c} = \left(\alpha^{-0(k-1)(q+1)}-1,\ldots, \alpha^{-(q^2-2)(k-1)(q+1)}-1,-1 \right).
\]
Since $s=\gcd(k-1,q-1)$, the zero components of $\mathbf{c}$ are $\alpha^{-i\frac{q-1}{s}(k-1)(q+1)}-1$ with $i\in[0,s(q+1)-1]$. By Theorem \ref{GRSMDSHULL},
${\rm Hull}_{\rm H}(GRS_k(\mathbf{b},\mathbf{a})) = GRS_{k-1}(\mathbf{b},\mathbf{a})$.
\end{proof}

\begin{thm}\label{GRScon4}
Given positive integers $k$ and $m$ with $k<q$ and $k-1<m<q-1$, let $\alpha$ be a primitive element of $\Ff_{q^2}$. Let
\[
s =\gcd(m-k+1,q-1) \mbox{ and } B=\left\{i+\frac{q-1}{s}j \, : \, i \in \left[\frac{q-1}{s}-1],j\in[0,(q+1)s-1 \right]\right\}.
\]
If $\mathbf{a}=((a_\ell)_{\ell\in B})$ and $\mathbf{b}=((\alpha^\ell)_{\ell\in B})$ are vectors of length $(q+1)(q-1-s)$, with $a_\ell^{q+1}=\alpha^{-\ell(k-1)(q+1)}-1$ for each $\ell\in B$, then
${\rm Hull}_{\rm H}(GRS_k(\mathbf{b},\mathbf{a})) = GRS_{k-1}(\mathbf{b},\mathbf{a})$.
\end{thm}
\begin{proof}
In (\ref{eqtr}), we put $\theta_{k-1,k-1}=1$, $\theta_{m,m}=-1$, and $\theta_{t,t}=\theta_{i,j}=0$ for each $t\in([k,q-1]\setminus\{m\})$ and each $(i,j)\in T$. The same argument as the one in the proof of Theorem \ref{GRScon3} leads us to the desired conclusion.
\end{proof}

The representation in \eqref{eqtr} provides a general technique to construct GRS codes whose Hermitian hulls are also MDS. To construct MDS EAQECCs with large minimum distances, we increase the dimension of GRS codes in the above four theorems, namely Theorems \ref{GRScon1} to \ref{GRScon4}, and obtain GRS codes whose Hermitian hulls contain MDS codes.

\begin{thm}\label{GRScon1-1}
Let $z$ and $k$ be positive integers such that $1\leq z < \lfloor q/2 \rfloor$ and $zq\leq k<(z+1)q-z-1$. Let $\alpha$ be a primitive element of $\Ff_{q^2}$. The code $GRS_k(\mathbf{b},\mathbf{1})$, of length $q^2$, with
\[
\mathbf{1} = (1,\ldots,1) \mbox{ and } \mathbf{b} =(\alpha^0,\ldots,\alpha^{q^2-2},0)
\]
has a $(k-z^2)$-dimensional Hermitian hull that contains $GRS_{q-1}(\mathbf{b},\mathbf{1})$ as a subcode.
\end{thm}

\begin{proof}
A generator matrix of $GRS_k(\mathbf{b},\mathbf{1})$ is
\[
G_k :=\begin{pmatrix}
\mathbf{v}_0 \\
\mathbf{v}_1\\
\vdots\\
\mathbf{v}_{k-1}
\end{pmatrix}
=\begin{pmatrix}
1 & 1 &\cdots & 1 & 1\\
1 & \alpha &\cdots & \alpha^{q^2-2}& 0\\
\vdots& \vdots &  \ddots & \vdots&\vdots\\
1 & \alpha^{k-1} &\cdots & \alpha^{(q^2-2)(k-1)}&0
\end{pmatrix}.
\]
We can then verify that $G_kG_k^{\dagger} = \left(\mathbf{v}_i\mathbf{v}_j^{\dagger} \right)_{i\in[0,k-1],j\in[0,k-1]}$ with
\[
\mathbf{v}_i\mathbf{v}_j^{\dagger}=
\begin{cases}
\sum_{\ell=0}^{q^2-1}\alpha^{(i+qj)\ell} & \mbox{if } i,j \in [1,k-1], \\
0& \mbox{if } i=0 \mbox{ or } j=0,
\end{cases}
\]
where $^\dag$ denotes the conjugate transpose.

Since, for each $i,j\in[1,k-1]$,
\[
\sum_{\ell=0}^{q^2-1}\alpha^{(i+qj)\ell}=
\begin{cases}
0 & \mbox{if } (q^2-1) \mbox{ does not divide } (i+qj),\\
-1  & \mbox{if } (q^2-1) \mbox{ divides } (i+qj),
\end{cases}
\]
we know that $(q^2-1)$ divides $(i+qj)$ if and only if $(i,j) \in \left\{(xq-y,yq-x):x,y\in[z]\right\}$. Let $\mathbf{e}_\ell$ be the unit vector of length $k$ with the property that the $\ell^{\rm th}$ coordinate is $1$ and the other coordinates are $0$. The set formed by the nonzero rows of $G_kG_k^{\dagger}$ is
\[
\{-\mathbf{e}_\ell:\ell=yq-x+1,x,y\in[z]\}.
\]
Hence, $\rank \left(G_kG_k^{\dagger}\right) = z^2$ and the Hermitian hull of $GRS_k(\mathbf{b},\mathbf{1})$ has dimension $k-z^2$. By Theorem \ref{GRScon1}, $GRS_{q-1}(\mathbf{b},\mathbf{1})$ is a subcode of the Hermitian hull of $GRS_k(\mathbf{b},\mathbf{1})$.
\end{proof}

\begin{thm}\label{GRScon2-1}
Let $z,f,k$ be positive integers such that
\[
1\leq z < \left \lfloor \frac{q^2-1}{2q} \right \rfloor, \, z+f+1<q, \, \mbox{ and }
zq\leq k<(z+1)q-z-f-1.
\]
Let $\alpha$ be a primitive element of $\Ff_{q^2}$. Given $GRS_k(\mathbf{b},\mathbf{a})$ of length $q^2-1$ with defining vectors
\[
\mathbf{a}=(\alpha^{-0(q-f-1)},\ldots,\alpha^{-(q^2-2)(q-f-1)}) \mbox{ and } \mathbf{b}=(\alpha^0,\ldots,\alpha^{q^2-2}),
\]
the Hermitian hull of $GRS_k(\mathbf{b},\mathbf{a})$ has dimension $k-z^2$ and contains $GRS_{q-f-1}(\mathbf{b},\mathbf{a})$ as a subcode.
\end{thm}
\begin{proof}
The same argument as the one in the proof of Theorem \ref{GRScon1-1} leads to the desired conclusion.
\end{proof}

\begin{thm}\label{GRScon3-1}
Let $z,f,k$ be positive integers such that
\[
1\leq z < \left \lfloor \frac{q^2-s(q+1)}{2q} \right \rfloor, \,
z+f+1<q, \mbox{ and }
zq \leq k <(z+1)q-z-f-1,
\]
with $s :=\gcd(q-f-1,q-1)$. Let $\alpha$ be a primitive element of $\Ff_{q^2}$. Let
\[
B:=\left\{i+\frac{q-1}{s} j : i \in \left[ \frac{q-1}{s}-1 \right],j \in[0,(q+1)s-1]\right\}.
\]
Let $\mathbf{a}=((a_\ell)_{\ell\in B},a)$ and $\mathbf{b}=((\alpha^\ell)_{\ell\in B},0)$ be vectors of length $q^2-s(q+1)$ with
\[
a^{q+1}=-1 \mbox{ and } a_\ell^{q+1} = \alpha^{-\ell(q-f-1)(q+1)}-1 \mbox{ for each } \ell\in B.
\]
The following statements hold.
\begin{enumerate}
\item If $f\geq z$, then $GRS_k(\mathbf{b},\mathbf{a})$ has a $(k-2z^2)$-dimensional Hermitian hull that contains  $GRS_{q-f-1}(\mathbf{b},\mathbf{a})$ as a subcode.
\item If $1\leq f < z$, then $GRS_k(\mathbf{b},\mathbf{a})$ has a $(k-z^2-zf)$-dimensional Hermitian hull that contains $GRS_{q-f-1}(\mathbf{b},\mathbf{a})$ as a subcode.
\end{enumerate}
\end{thm}

\begin{proof}
Let $G_k$ be a generator matrix of $GRS_k(\mathbf{b},\mathbf{a})$ in the form of \eqref{GRSgenerator} and let $\mathbf{v}_0,\mathbf{v}_1,\cdots,\mathbf{v}_{k-1}$ be the rows of $G_k$. Hence,
\[
G_k G_k^{\dagger} = \left(\mathbf{v}_i\mathbf{v}_j^{\dagger}\right)_{i\in[0,k-1],j\in[0,k-1]},
\]
with
\[
\mathbf{v}_i\mathbf{v}_j^{\dagger} =
\begin{cases}
0 & \mbox{if } i=j=0,\\
\sum_{\ell\in B} \alpha^{(i-q+f+1+q(j-q+f+1))\ell} -
\sum_{\ell\in B}\alpha^{(i+qj)\ell} & \mbox{otherwise.}
\end{cases}
\]
Let $g_{i,j}$ stand for $i-q+f+1+q(j-q+f+1)$. Since $1\leq z < \lfloor q/2 \rfloor$ and $z+f+1<q$, we know that $(q^2-1)$ cannot simultaneously divide both $g_{i,j}$ and $(i+qj)$. For completeness, we compute $\mathbf{v}_i\mathbf{v}_j^{\dagger}$ in three cases.
\begin{enumerate}[wide, itemsep=0pt, leftmargin =0pt, widest={{\bf Case $2$}}]
\item[{\bf Case $1$}:] If $(q^2-1)$ divides neither $g_{i,j}$ nor $(i+qj)$, then
\[
\mathbf{v}_i\mathbf{v}_j^{\dagger} = -\sum_{\ell=0}^{s(q+1)-1}\alpha^{g_{i,j}\ell\frac{q-1}{s}} + \sum_{\ell=0}^{s(q+1)-1}\alpha^{(i+qj)\ell\frac{q-1}{s}},
\]
since
\[
\sum_{\ell=0}^{q^2-1}\alpha^{g_{i,j}\ell} =\sum_{\ell=0}^{q^2-1}\alpha^{(i+qj)\ell} = 0.
\]
Recalling that $s=\gcd(q-f-1,q-1)$, we infer that $s(q+1)$ divides $g_{i,j}$ if and only if $s(q+1)$ divides $(i+qj)$. Thus, $\mathbf{v}_i\mathbf{v}_j^{\dagger} = 0$ regardless of whether $s(q+1)$ divides $g_{i,j}$ or not.

\item[{\bf Case $2$}:] If $(q^2-1)$ divides $g_{i,j}$ but does not divide  $(i+qj)$, then
\[
\mathbf{v}_i\mathbf{v}_j^{\dagger} = |B| + \sum_{\ell=0}^{s(q+1)-1} \alpha^{(i+qj)\ell\frac{q-1}{s}}.
\]
Since $(q^2-1)$ divides $g_{i,j}$, we know that $s(q+1)$ divides $g_{i,j}$, which implies that $s(q+1)$ divides $(i+qj)$, giving us  $\mathbf{v}_i\mathbf{v}_j^{\dagger}=-1$.

\item[{\bf Case $3$}:] If $(q^2-1)$ does not divide $g_{i,j}$ but divides $(i+qj)$, then, using a similar method as in {\bf Case $2$}, we get $\mathbf{v}_i\mathbf{v}_j^{\dagger}=1$.
\end{enumerate}
		
The proof of Theorem \ref{GRScon1-1} tells us that $(q^2-1)$ divides $(i+qj)$ if and only if $(i,j) \in \left\{(xq-y,yq-x):x,y\in[z]\right\}$ for each $(i,j)\in\left([0,k-1]\times[0,k-1]\setminus\{(0,0)\}\right)$. Since $z+f+1<q$ and $zq\leq k<(z+1)q-z-f-1$, we conclude that $(q^2-1)$ divides $g_{i,j}$ if and only if
\[
(i,j)\in\left\{((x+1)q-y-f-1,(y+1)q-x-f-1):x,y\in[0,z-1]\right\}
\]
for each
\[
(i,j)\in\left([0,k-1]\times[0,k-1]\setminus\{(0,0)\}\right).
\]
Let $\mathbf{e}_\ell$ be the unit vector of length $k$ such that its $\ell^{\rm th}$ coordinate is $1$ and the other coordinates are $0$.
		
If $f\geq z$, then the set formed by the nonzero rows of $G_kG_k^{\dagger}$ is
\[
\{\mathbf{e}_\ell: \ell=yq-x+1,x,y\in[z]\} \bigcup \{-\mathbf{e}_\ell:\ell=(y+1)q-x-f,x,y\in[0,z-1]\}.
\]
Hence, $\rank\left(G_kG_k^{\dagger}\right) = 2z^2$ and the dimension of the Hermitian hull of $GRS_k(\mathbf{b},\mathbf{a})$ is $k-2z^2$. By Theorem \ref{GRScon3}, we can conclude that $GRS_{q-f-1}(\mathbf{b},\mathbf{a})$ is a subcode of the Hermitian hull of $GRS_k(\mathbf{b},\mathbf{a})$.
		
If $1\leq f < z$, then the set formed by the nonzero rows of $G_kG_k^{\dagger}$ is the union of three sets
\begin{align*}
&\{\mathbf{e}_\ell:\ell=yq-x+1,x\in[z],y\in[f]\},\\
&\{-\mathbf{e}_\ell:\ell=(y+1)q-x-f,x\in[0,z-1],y\in[z-f,z-1]\},\\
&\{\mathbf{e}_{\ell_1}-\mathbf{e}_{\ell_2}: \ell_1=yq-x+1,\ell_2=(y-f)q-x-f+1,x\in[z],y\in[f+1,z]\}.
\end{align*}
Hence, $\rank\left(G_kG_k^{\dagger}\right) = z^2+zf$ and the Hermitian hull of $GRS_k(\mathbf{b},\mathbf{a})$ has dimension $k-z^2-zf$. By Theorem \ref{GRScon3}, $GRS_{q-f-1}(\mathbf{b},\mathbf{a})$ is a subcode of the Hermitian hull of $GRS_k(\mathbf{b},\mathbf{a})$.
\end{proof}

\begin{thm}\label{GRScon4-1}
Let $m$ and $f$ be positive integers such that $q-f-1<m<q-1$ and let $s :=\gcd(m-q+f+1,q-1)$. Let $z$ and $k$ be positive integers satisfying
\[
1\leq z < \left \lfloor \frac{(q+1)(q-1-s)}{2q} \right \rfloor, \,
z+f+1<q, \mbox{ and } zq \leq k<(z+1)q-z-f-1.
\]
Let $\alpha$ be a primitive element of $\Ff_{q^2}$. Let
\[
B=\left\{i+\frac{q-1}{s}j : i\in \left[\frac{q-1}{s}-1 \right] , j\in[0,(q+1)s-1]\right\}.
\]
If $\mathbf{a}=(a_\ell)_{\ell\in B}$ and $\mathbf{b}=(\alpha^\ell)_{\ell\in B}$ are vectors of length $(q+1)(q-1-s)$, with $a_\ell^{q+1}=\alpha^{-\ell(q-f-1)(q+1)}-1$ for each $\ell\in B$, then the following statements hold.
\begin{enumerate}
\item If $f\geq z$, then $GRS_k(\mathbf{b},\mathbf{a})$ has a $(k-2z^2)$-dimensional Hermitian hull that contains  $GRS_{q-f-1}(\mathbf{b},\mathbf{a})$ as a subcode.
\item If $1\leq f < z$, then $GRS_k(\mathbf{b},\mathbf{a})$ has a $(k-z^2-zf)$-dimensional Hermitian hull that contains $GRS_{q-f-1}(\mathbf{b},\mathbf{a})$ as a subcode.
\end{enumerate}
\end{thm}

\begin{proof}
The same argument as the one in the proof of Theorem \ref{GRScon3-1} leads to the desired conclusion.
\end{proof}

\subsection{Hermitian MDS Hulls from Two-Point Rational AG Codes}\label{subsec:non-GRS typed}

In this subsection, we construct MDS linear codes, whose Hermitian hulls are MDS, from two-point rational AG codes.

\begin{thm} \label{thm:rankcomp}
Let $D = P_1+\ldots+P_n$ and $G = k \, O$, with $0\leq k\leq
\left\lfloor \frac{n-2}{q+1} \right\rfloor$. Let $H$ be a divisor and let $\omega$ be a differential form with the properties that
\begin{itemize}
	\item $(\omega) = G + H - D$ and
	\item for any $1 \leq i \leq n$ and for some $a_i\in \Ff_{q^2}^*$, $ v_{P_i}(\omega)=-1$ and
	$\Res_{P_i}(\omega) = a_i^{q+1}$.
\end{itemize}
Let $P$ be a rational place such that $P\notin \textnormal{supp}(D)\cup \textnormal{supp}(G)$. Set $G'=G+P_0$ and $H'=D-G'+(\omega)$, and assume that $\deg (G'\vee H')<n$. Then the following statements hold.
\begin{enumerate}
\item If ${\bf a} \cdot \cc_{{\cal L}}(D, P)$ is a subcode of
$({\bf a} \cdot \cc_{{\cal L}}(D, G'))^{\perp_{\rm H}}$, then ${\bf a}\cdot \cc_{{\cal L}}(D, G')$ is Hermitian self-orthogonal with parameters $[n,k+2, n-k-1]_{q^2}$.

\item If ${\bf a} \cdot \cc_{{\cal L}}(D, P)$ is not a subcode of $({\bf a} \cdot \cc_{{\cal L}}(D, G'))^{\perp_{\rm H}}$, then ${\bf a} \cdot \cc_{{\cal L}}(D, G')$ has parameters $[n,k+2, n-k-1]_{q^2}$. Its Hermitian hull is a $k$-dimensional MDS code.
\end{enumerate}
\end{thm}

\begin{proof}
We choose a subset $U \subseteq \Ff_{q^2}$ of cardinality $n$ and use it to define
\[
h(x) = \prod\limits_{\alpha\in U}(x-\alpha) \mbox{ and }
D=(h)_0.
\]
We then use
\[
\omega=\frac{dx}{h(x)}, \quad H=D-G+(\omega), \quad H'=D-G'+(\omega).
\]
It is immediate to confirm that
\[
H'=(n-k-2)O-P, \quad
G'\wedge H'= k \, O-P, \quad
\deg (G'\vee H')=(n-k-2) \, O+P.
\]
Let ${\bf b}=(b_1,\ldots,b_n)$ with $b_i=a_i^{\frac{q+1}{2}}$. Let
\[
\cc_0={\bf b} \cdot \cc_{\cal L}(D, G), \quad
\cc_0'={\bf b} \cdot \cc_{\cal L}(D, G'), \quad
\cc={\bf a} \cdot \cc_{\cal L}(D, G), \quad
\cc'={\bf a} \cdot \cc_{\cal L}(D, G').
\]
Under the condition that $\deg (G'\vee H')< n$, we use a similar technique as that in the proof of \cite[Lemma 5]{Sok2022b} to confirm that
\[
\cc_{{\cal L}}(D,G') \cap  \cc_{{\cal L}}(D,H') =
\cc_{{\cal L}}(D,G'\wedge H').
\]
From the fact that $H'+G'-D=(H'\vee G')+(H'\wedge G')-D$, we get  $(\omega)-(H'\wedge G')=(H'\vee G')-D$ and, hence,
$\dim({\cal L}((\omega)-(H'\wedge G')))=0$ and
\begin{multline*}
\dim(\cc_{{\cal L}}(D,G'\wedge H')) =
\dim({\cal L}(H'\wedge G'))-
\dim({\cal L}((H'\wedge G')-D))\\
= \deg(H'\wedge G') + 1 + \dim( {\cal L}((\omega)-(H'\wedge G')))=k.
\end{multline*}
The first equality follows from Proposition \ref{thm:distance}. The second equality is due to Proposition \ref{thm:Riemann-Roch}.
Since
\[
\left({\bf b} \cdot \cc_{\cal L}(D, G')\right)^{\perp_{\rm E}} =
\left(b_1^{-1},\ldots,b_n^{-1}\right) \cdot \cc_{\cal L}(D, G')^{\perp_{\rm E}} = {\bf b}\cdot \cc_{\cal L}(D, H'),
\]
we can use Lemma \ref{lem:dual2} to infer that
\[
\cc_0' \cap {\cc_0'}^{\perp_{\rm E}} =
{\bf b} \cdot (\cc_{\cal L}(D,G') \cap \cc_{\cal L}(D, H')) =
{\bf b} \cdot \cc_{\cal L}(D,G'\wedge H'),
\]
which is MDS with dimension $k$. Since $k \leq \left \lfloor \frac{n-2}{q+1} \right \rfloor$ and $\Res_{P_i}(\omega) \in \Ff_q$ for any $1\le i\le n$, we obtain $\cc \subset \cc^{\perp_{\rm H}}$ by \cite[Theorem 3]{SokQSC}.
	
Writing the respective generator matrices of $\cc$ and $\cc'$ as ${\cal G}=({\bf g}_1,{\bf g}_2,\cdots,{\bf g}_k)^\top$ and ${\cal G}'=({\bf g}_1,{\bf g}_2,\cdots,{\bf g}_k,{\bf x})^\top$, we have
\[
{\cal G}'{{\cal G}'}^\dag =
\left(
\begin{array}{ccc|c}
& & & {\bf g}_1{\bf x}^\dag\\
& {\cal G}{{\cal G}}^\dag & & \vdots\\
& & & {\bf g}_k{\bf x}^\dag\\
\hline
{\bf x}{\bf g}_1^\dag & \cdots & {\bf x}{\bf g}_k^\dag &
{\bf x}{\bf x}^\dag\\
\end{array}
\right).
\]
If ${\bf a} \cdot \cc_{{\cal L}}(D, P)$ is a subcode of $\cc^{\perp_{\rm H}}$, then $\cc'$ is Hermitian self-orthogonal. If ${\bf a} \cdot \cc_{{\cal L}}(D, P)$ is not a subcode of $\cc^{\perp_{\rm H}}$, then
$\rank({\cal G}'{{\cal G}'}^\dag)= 2$, since
$\rank({\cal G}{{\cal G}}^\dag)=0$ and
${\bf g}_i{\bf x}^\dag \neq 0$ for some $1\le i\le k$.
By \cite[Proposition 3.2]{Guenda2017}, ${\bf a} \cdot
\cc_{{\cal L}}(D, G')$ has Hermitian hull of dimension $k-1$. To show that its Hermitian hull is MDS, we verify that
$\cc_0' \cap {\cc_0'}^{\perp_{\rm E}}$ and $\cc' \cap
\cc^{\perp_{\rm H}} = {\bf a} \cdot (\cc_{\cal L}(D,G') \cap
\cc_{\cal L}(D, H')^q)$ are equivalent codes. To do so, we refer to \cite[Proposition 11]{Pereira2021} for the isomorphism
$\cc_{\cal L}(D, H') \longrightarrow \cc_{\cal L}(D, H')^q$ that maps ${\bf x}_i \mapsto {\bf x}_i^q$ for $1 \leq i \leq n-k-2$, with $\{{\bf x}_1,\cdots, {\bf x}_{n-k-2}\}$ being a basis for
$\cc_{\cal L}(D, H')$.
\end{proof}

\begin{remark} Note that starting from a GRS code which is Hermitian self-orthogonal, the choice of the place $P$ such that ${\bf a}\cdot C_{{\cal L}}(D, P)$ is not a subcode of $({\bf a}\cdot C_{{\cal L}}(D, G'))^{\perp_{\rm H}}$ is always guaranteed if the set of evaluation points, used to construct such a Hermitian self-orthogonal code, is not equal to $\Ff_{q^2}$. A basis of ${\bf a}\cdot C_{{\cal L}}(D, P)$ can be obtained from choosing a suitable rational function which has a pole at $P$.
\end{remark}

The following corollary gives our first construction.
\begin{cor}\label{cor:1}
Let $q$ be a prime power. If $(s-1)$ is a divisor of $(q^2-1)$, with $s \neq q^2$, then,
for $0\le k\le \lfloor \frac{s-2}{q+1}\rfloor$, there exists an $[s,k+2, s-k-1]_{q^2}$ code $\cc_0$ whose Hermitian hull is MDS with dimension $k$.
\end{cor}

\begin{proof}
Let $U=\left\{\alpha\in \F_{q^2} : \alpha^{s-1}=1\right\} \cup \{0\}$. Writing the elements more explicitly, we have $U=\{u_1,\hdots,u_s\}$, with $u_s=0$.
Let
\[
h(x)=\prod\limits_{\alpha\in U}(x-\alpha), \quad
D=(h)_0, \quad \omega=\frac{dx}{h(x)}, \quad G=k O+P.
\]
Taking the derivative of $h(x)$ yields $h'(x)=s \, x^{s-1}-1$. Hence, for any $1 \leq i \leq s-1$, we know that $h'(u_i) = s-1 \in \F_{q}^*$ and $h'(u_s)=-1 \in \F_{q}^*$. Thus, for any $1\leq i\leq s$, we have $h'(u_i)=\beta_i^{q+1}$, for some $\beta_i\in \F_{q^2}$. We are guaranteed the existence of the required rational place $P$ since $s < q^2$.
Taking $H=D-G+(\omega)=(s-k-2) O-P$ gives us
$G\wedge H=k O-P$ and $G\vee H=(s-k-2) O+P.$
It is then clear that $\deg (G\wedge H)=k-1$, and $\deg (G\vee H)<s$.
The conclusion of the corollary now follows from Theorem \ref{thm:rankcomp}.
\end{proof}

Next, we construct linear codes with $k$-dimensional Hermitian hull, for which the set of evaluation points is a union of additive cosets of $\Ff_q$ in $\Ff_{q^2}$.
\begin{cor}\label{cor:2}
Let $q$ be a prime power. If $t$ is an integer such that  $1 \leq t < q$, then, for $0\le k \le \lfloor \frac{tq-2}{q+1}\rfloor$, there exists a $[tq,k+2, tq-k-1]_{q^2}$ code $\cc_1$ whose Hermitian hull is MDS with dimension $k$.
\end{cor}
\begin{proof}
Fixing an element $\alpha\in \Ff_{q^2}\setminus \Ff_q$, we write $\Ff_q = \{u_1,\ldots, u_q\}$ and define $\alpha_{i,j} := u_i \, \alpha + u_j$, for $1\le i\le t$ and $1\le j\le q$. The set $U=\{\alpha_{i,j} \, : \, 1\le i\le t, 1\le j\le q\}$ clearly has $tq$ elements.
Let
\[
h(x)=\prod\limits_{\alpha\in U}(x-\alpha), \quad
D=(h)_0, \quad \omega=\frac{dx}{h(x)}, \quad
G=k O+P.
\]
Performing the computation as the one done in \cite[Construction 4]{SokQSC}, we obtain, for any $\alpha\in U$, $h'(\alpha)=\beta^{q+1}$ for some $\beta\in \Ff_{q^2}.$

We select $H=D-G+(\omega)=(tq-k-2) O-P$ to arrive at
$G\wedge H=k O-P$ and $G\vee H=(tq-k-2) O+P.$
The conclusion of the corollary now follows from Theorem \ref{thm:rankcomp}.
\end{proof}

\begin{example} Over $\F_{25}$, there are $26$ rational places, and we consider the following $22$ places, with $\theta$ being the standard primitive element of $\F_{25}$ in {\tt MAGMA}:
\begin{align*}
O &=(1 : 0 : 0) &  P &=(\theta^{10} : 1 : 0) & P_1 & =(0 : 1 : 0) &
P_2 &=(\theta : 1 : 0) & P_3 &=(\theta^2 : 1 : 0) \\
P_4 &= (\theta^3 : 1 :0) & P_5 &= (\theta^4 : 1 : 0) &
P_6 &= (\theta^5 : 1 : 0) & P_7 &= (2 : 1 : 0) & P_8 &= (\theta^7 : 1 : 0)\\
P_9&=(\theta^8 : 1 : 0) & P_{10} &= (\theta^9 : 1 : 0) &
P_{11}&=(\theta^{11} : 1 : 0) & P_{12}&= (\theta^{13} : 1 : 0) &
P_{13} &= (\theta^{14} : 1 : 0)\\
P_{14}&=(\theta^{15} : 1 : 0) & P_{15}&=(\theta^{16} : 1 : 0) &
P_{16}&=(3 : 1 : 0)& P_{17}&= (\theta^{19} : 1 : 0) &
P_{18}&= (\theta^{22} : 1 : 0) \\
P_{19}&= (\theta^{23} : 1 : 0) & P_{20}&=(1 : 1 : 0). &&&&&&
\end{align*}
Taking $D=P_1+\ldots+P_{20}$ and $G=3 O+P$, the residues of $\omega $ at $(P_i)_{1\le i\le 20}$ are
\[
{\rm Res}(\omega)=
\left( \theta^{22}, \theta^{22}, \theta^{21}, \theta^{21}, 1,1, \theta^{22}, \theta^{23}, 1, \theta^{23}, \theta^{22}, 1,1, \theta^{21}, \theta^{23}, \theta^{21}, \theta^{22}, \theta^{23}, \theta^{23}, \theta^{21}\right).
\]
We then obtain the code ${\rm Res}(\omega) \cdot C_{{\cal L}}
(D,3 O+P)$, whose Hermitian hull is MDS with dimension $3$. The code has a generator matrix ${\cal G}=({\cal G}_0~{\cal G}_1)$, where
\begin{align*}
{\cal G}_0 &=
\begin{pmatrix}
\theta^{22}&\theta^{22}&\theta^{21}&\theta^{21}&1&1&\theta^{22}&\theta^{23}&1&\theta^{23}\\
0&\theta^{11}&\theta^{11}&4&\theta^{16}&3&\theta^{17}&\theta^{21}&\theta^{23}&\theta^{23}\\
0&1&\theta&\theta^{3}&\theta^{8}&4&4&\theta^{19}&\theta^{22}&\theta^{23}\\
0&\theta^{13}&\theta^{15}&3&1&2&\theta^{7}&\theta^{17}&\theta^{21}&\theta^{23}\\
1&\theta^{22}&\theta^{3}&\theta^{22}&\theta^{14}&\theta^{7}&\theta^{17}&\theta^{19}&\theta^{4}&\theta^{8}\\
\end{pmatrix} \mbox{ and} \\
{\cal G}_1 &=
\begin{pmatrix}
\theta^{22}&1&1&\theta^{21}&\theta^{23}&\theta^{21}&\theta^{22}&\theta^{23}&\theta^{23}&\theta^{21}\\
\theta^{23}&\theta^{2}&\theta^{3}&\theta&\theta^{4}&\theta^{3}&\theta^{5}&\theta^{7}&\theta^{8}&\theta^{8}\\
1&\theta^{4}&2&\theta^{5}&\theta^{9}&\theta^{9}&4&\theta^{15}&\theta^{17}&\theta^{19}\\
\theta&2&\theta^{9}&\theta^{9}&\theta^{14}&\theta^{15}&\theta^{19}&\theta^{23}&\theta^{2}&2\\
\theta^{14}&\theta^{3}&4&\theta^{5}&\theta^{17}&\theta^{19}&\theta^{13}&4&\theta^{9}&3\\
\end{pmatrix}.
\end{align*}
\end{example}

The next construction utilizes a set of evaluation points which is a union of multiplicative cosets in $\F_{q^2}$.

\begin{cor}\label{cor:3}
Let $n_0$ be a divisor of $(q^{2}-1)$ and let $n_{2} =
\frac{n_0}{\gcd (n_0,q+1)}$. If $t$ is an integer such that $1 \leq t \leq \frac{q-1}{n_2}-2$, then, for $0 \leq k \leq
\left \lfloor \frac{(t+1)n_0-1}{q+1} \right\rfloor$, there exists a $[(t+1)n_0+1,k+2, (t+1)n_0-k]_{q^2}$ code $\cc_2$ whose Hermitian hull is MDS with dimension $k$.
\end{cor}
\begin{proof} We write $U_{n_0}=\{u\in \F_{q^2} \, : \, u^{n_0}=1\}$. For $1 \leq t < \frac{q-1} {n_{2}}$, let
\[
h(x) =\prod\limits_{\alpha\in U}(x-\alpha) \mbox{ and }
U = \{a_1,\hdots,a_{(t+1)n_0+1}\} = U_{n_0} \bigcup \left(\bigcup^{t}\limits_{j=1} \alpha_{j} U_{n_0} \right) \bigcup \{0\},
\]
with $(\alpha_{j})_{1\le j\le t}$ being elements in the distinct multiplicative cosets of $U_{n_0}$. The cardinality of $U$ is $(t+1)n_0+1$. Let
\[
h(x)=\prod\limits_{\alpha\in U}(x-\alpha), \quad
D=(h)_0, \quad \omega=\frac{dx}{h(x)}, \quad G=k O+P.
\]
Computing as in \cite[Construction 3]{SokQSC}, for any $\alpha\in U$, we get $h'(\alpha)=\beta^{q+1}$ for some $\beta\in \Ff_{q^2}$. Choosing
\[
H=D-G+(\omega)=((t+1)n_0-k) \, O - P
\]
leads to
\[
G\wedge H = k O-P \mbox{ and } G\vee H=((t+1)n_0-k) O+P.
\]
The desired conclusion now follows from Theorem \ref{thm:rankcomp}.
\end{proof}

The next theorem constructs an MDS linear code, whose Hermitian hull is MDS, from an extended MDS Hermitian self-orthogonal code, if the latter exists.
\begin{thm}\label{thm:MDSHull}
Let $D = P_1 + \ldots + P_n$ and $G=kO$ with $0 \le k \le \left\lfloor \frac{n-2}{q+1} \right\rfloor$. We denote the extended code of $\cc_{{\cal L}}(D, G)$ by $\bar{\cc}_{{\cal L}}(D, G)$. Let $\omega$ be a differential form such that $(\omega) = G + H - D$ for some divisor $H$ and let $1 \leq v_{P_i}(\omega) = -1$ for $1 \le i \le n$. Let ${\bf a} \cdot \bar{\cc}_{{\cal L}}(D, G)$ be a Hermitian self-orthogonal code with parameters $[n+1,k+1,n-k+1]_{q^2}$, where ${\bf a}=(a_1,\ldots, a_n,0)$ with $a_i^{q+1} = \Res_{P_i}(\omega)$. Let $P$ be a rational place such that $P\notin \textnormal{supp}(D)\cup \textnormal{supp}(G)$. We set $G'=G+P_0$ and $H'=D-G'+(\omega)$, and assume that $\deg (G'\vee H')<n$. Then the following assertions hold.
\begin{enumerate}
\item If ${\bf a}\cdot \bar{\cc}_{{\cal L}}(D, P)$ is a subcode of $({\bf a}\cdot \bar{\cc}_{{\cal L}}(D, G'))^{\perp_{\rm H}}$, then ${\bf a}\cdot \bar{\cc}_{{\cal L}}(D, G')$ is an $[n+1,k+2, n-k]_{q^2}$-Hermitian self-orthogonal code.
\item If ${\bf a}\cdot \bar{\cc}_{{\cal L}}(D, P)$ is not a subcode of $({\bf a}\cdot \bar{\cc}_{{\cal L}}(D, G'))^{\perp_{\rm H}}$, then ${\bf a}\cdot \bar{\cc}_{{\cal L}}(D, G')$ has parameters $[n+1,k+2, n-k]_{q^2}$ and its Hermitian hull is MDS with dimension $k$.
\end{enumerate}
\end{thm}
\begin{proof}
With the same setting as that in the proof of Theorem \ref{thm:rankcomp}, we verify that
\[
H'=(n-k-2)O-P, \quad G'\wedge H'=kO-P, \quad \deg (G'\vee H') \le n-1.
\]
We select $b_i=a_i^{\frac{q+1}{2}}$ for $1\le i\le n$ to build
${\bf b}=(b_1,\ldots,b_n,0)$. Let
\[
\cc_0={\bf b}\cdot \cc_{\cal L}(D, G), \quad
\cc_0'={\bf b}\cdot \cc_{\cal L}(D, G'), \quad
\cc={\bf a}\cdot \cc_{\cal L}(D, G), \quad
\cc'={\bf a}\cdot \cc_{\cal L}(D, G').
\]
We note that $\cc' \cap {\cc'}^{\perp_{\rm E}} =
\bar{\cc}_{\cal L}(D,G'\wedge H')$, which is an MDS code with dimension $k$. The rest of the proof follows from the same reasoning as that in the proof of Theorem \ref{thm:rankcomp}.
\end{proof}

Combining Theorem \ref{thm:rankcomp} and Corollaries \ref{cor:1}-\ref{cor:3}, we now elaborate on how to construct an MDS linear code whose Hermitian hull is \emph{not} necessarily MDS.

\begin{thm}\label{thm:arbitrary-hull} We assume the notation as in Theorem \ref{thm:rankcomp}. An $[n,k+2]_{q^2}$ MDS code, with $0\le k\le \left\lfloor \frac{n-2}{q+1} \right\rfloor$, has a Hermitian hull of dimension $\ell \in \{0,1,\ldots, k\}$ if one of the following conditions holds.
\begin{enumerate}
\item $(n-1)$ divides $(q^2-1)$ for $n \neq q^2$,
\item $n=tq$ for $1\le t\le q-1$,
\item $n=(t+1)n_0+1$ with $n_0$ being a divisor of $(q^{2}-1)$, $n_{2}=\frac{n_0}{\gcd (n_0,q+1)}$, and $1\le t\le \frac{q-1}{n_2}-2$.
\end{enumerate}
\end{thm}

\begin{proof} Consider the algebraic geometry code
$\cc_{\cal L}(D,G)$ with $G=k O+P$. If any one of the three conditions holds and if $\Res_P(\omega)\in \F_{q}^n$, then $\cc_0=C_{\cal L}(D,G_0)$, with $G_0=k O$, is Hermitian self-orthogonal. Let ${\cal G}_0$ and ${\cal G}$ be generator matrices of $\cc_0$ and $\cc$, respectively, with
${\cal G}=
\begin{pmatrix}
{\cal G}_0\\
{\bf x}
\end{pmatrix}$.
Then, up to equivalence, we can write ${\cal G}$ as
${\cal G}=
\begin{pmatrix}
\diag(1,\ldots,1)~A\\
{\bf x}
\end{pmatrix}$. For $\alpha\in \F_{q^2}$ such that $\alpha^{q+1} \neq -1$, take
\[
{\cal G}^\ell=
\begin{pmatrix}
\diag(1,\cdots,1,\underbrace{\alpha,\cdots,\alpha}\limits_{\ell})~A\\
{\bf x}
\end{pmatrix}.
\]
It follows that, for any $\ell$ such that $0\le \ell \le k$, ${\cal G}$ and ${\cal G}^{(\ell)}$ generate linear codes with the same parameters. Since $\cc_0$ is Hermitian self-orthogonal,
\[
\rank\left({\cal G}^{(\ell)}{{\cal G}^{(\ell)}}^\dag\right) = \ell+1.
\]
Thus, the Hermitian hull dimension of $\cc$ is $(k+1)-(\ell+1)$.
\end{proof}

The following proposition allows us to construct a larger set of evaluation points from a smaller one such that the larger set produces the differential form $\omega$ with properties that satisfy the conditions in Theorems \ref{thm:rankcomp} and \ref{thm:MDSHull}. In other words, for any $1 \leq i \leq n$, we have $v_{P_i}(\omega)=-1$ and $\Res_{P_i}(\omega) = a_i^{q+1}$ for some $a_i\in \Ff_{q^2}^*$.

\begin{prop}\label{prop:recursive-1}
Let $S^{(q+1)} := \{ \gamma^{q+1} : \gamma \in \F_{q^2}\setminus \{0\}\}$. Given a set $U$, let $f(x)=\prod\limits_{\alpha\in U} (x-\alpha)$ be such that $f'(\alpha)\in S^{(q+1)}$ for any $\alpha \in U$. Let there be $\beta_1,\beta_2 \in \F_{q^2}\setminus U$ with the following properties:
\begin{enumerate}[1)]
\item $S^{(q+1)}$ contains $f(\beta_1)(\beta_1-\beta_2)$ and $f(\beta_2)(\beta_2-\beta_1)$.
\item $\{(\alpha-\beta_1)(\alpha-\beta_2) : \alpha\in U \} \subseteq S^{(q+1)}$.
\end{enumerate}
Then there exist $V \subseteq \F_{q^2}$ and a polynomial $g(x)$ of degree $|V|= |U|+2$ such that $g'(\beta)\in S^{(q+1)}$ for any $\beta \in V$.
\end{prop}
\begin{proof}
Assuming such $\beta_1$ and $\beta_2$ exist, we let $V =U \cup \{\beta_1,\beta_2\}$ and use $g(x)= \prod\limits_{\beta\in V} (x-\beta)$. It is clear that $\deg(g(x))=|V|$. Taking the derivative of $g(x)$ yields
\[
g'(x)=f'(x)(x-\beta_1)(x-\beta_2)+f(x)((x-\beta_1)+(x-\beta_2)).
\]
We can then infer that $S^{(q+1)}$ contains $g'(\beta_1)$ and $g'(\beta_2)$ by Property 1) and that $\{g'(\alpha) : \alpha \in U\} \subseteq S^{(q+1)}$ by Property 2), respectively.
We conclude that $g'(\beta) \in S^{(q+1)}$ for any $\beta \in V$.
\end{proof}

\begin{example}
Let $\theta$ be the usual primitive element of $\F_{q^2}$ given by \texttt{MAGMA} for $q \in \{5,7,9,11\}$.
\begin{itemize}
\item When $q=5$, starting from $f(x) = g_0(x) = \prod\limits_{\beta\in \F_5}(x-\beta)$, we can recursively construct the polynomials $g_i(x)=\prod\limits_{\beta\in U_i}(x-\beta)$ of respective degrees $7$ and $9$, with
$U_1 = \{\underbrace{0, 1, 2, 3, 4}, \underbrace{\theta, \theta^5}  \}$ and
$U_2=\{\underbrace{0, 1, 2, 3, 4}, \underbrace{\theta, \theta^5}, \underbrace{\theta^4, \theta^{20}} \}$. The two sets of evaluation points give rise to two MDS two-point AG codes of lengths $7$ and $9$ over $\F_{25}$.

\item When $q=7$, starting from $f(x) = g_0(x)=\prod\limits_{\beta\in \F_7}(x-\beta)$, we can construct a polynomial $g(x)=\prod\limits_{\beta\in V}(x-\beta)$ of degree $13$ with
$V=\{ \underbrace{0, 1, 2, 3, 4, 5, 6}, \underbrace{\theta, \theta^7}, \underbrace{\theta^{27}, \theta^{45}}, \underbrace{\theta^{34}, \theta^{46}} \}$. The code of length $11$ over $\Ff_{49}$ constructed from $V$ is not covered by the code families in Corollaries \ref{cor:1} to \ref{cor:3}.
\item When $q=9$, we can construct a polynomial $g(x) = \prod\limits_{\beta\in V}(x-\beta)$ of degree $17$, with
\[
V=\{ \underbrace{1, \theta^{10}, \theta^{20}, \theta^{30}, 2, \theta^{50}, \theta^{60}, \theta^{70}, 0}, \underbrace{\theta, \theta^9}, \underbrace{\theta^8, \theta^{72}}, \underbrace{\theta^{23}, \theta^{47}},\underbrace{\theta^{44}, \theta^{76}}  \}.
\]
The code of length $13$ over $\Ff_{81}$, constructed using a subset of $V$, is not covered by the code families in Corollaries \ref{cor:1} to \ref{cor:3}.
\item When $q=11$, we can construct a polynomial $g(x)=\prod\limits_{\beta\in V}(x-\beta)$ of degree $21$, with
\[
V=\{\underbrace{0, 1, 2, 3, 4, 5, 6, 7, 8, 9, 10}, \underbrace{\theta, \theta^{11}}, \underbrace{\theta^4, \theta^{44}}, \underbrace{\theta^{17}, \theta^{67}}, \underbrace{\theta^{26}, \theta^{46}},\underbrace{\theta^{63}, \theta^{93}}\}.
\]
The code of length $17$ over $\Ff_{121}$, constructed from a subset of $V$, is not contained in the code families in Corollaries \ref{cor:1} to \ref{cor:3}.
\end{itemize}
\end{example}

\section{Quantum error correction}\label{sec:4}
	
A $q$ ary \emph{quantum error-correcting code} (QECC), also known as a {\it qudit code}, is a $K$-dimensional subspace of $(\mathbb{C}^q)^{\otimes n}$. We use the respective terms {\it qubit} and {\it qutrit} codes when $q=2$ and $q=3$. The parameters $[[n,\kappa,\delta]]_q$ of a QECC signifies that the code has dimension $q^{\kappa}$ and can correct quantum error operators affecting up to $\lfloor(\delta-1)/2\rfloor$ arbitrary positions in the quantum ensemble. Formalizing the stabilizer framework, first introduced by Gottesman in \cite{Gottesman1997}, Calderbank {\it et al}. in \cite{Calderbank1998} proposed a general approach to construct qubit QECCs. This method was subsequently extended to the nonbinary case in \cite{Ketkar2006}, establishing the correspondence between a Hermitian self-orthogonal classical code and a stabilizer QECC.
	
The stabilizer frameworks requires the ingredient classical additive codes to be trace-Hermitian self-orthogonal. When the codes are linear, the requirement translates to Hermitian self-orthogonality.  One can relax the orthogonality condition while still being able to perform quantum error control in the frameworks of \emph{entanglement-assisted quantum error-correcting codes} (EAQECCs).
	
Brun, Devetak, and Hsieh introduced EAQECCs in \cite{Brun2006}. The sender and the receiver share pairs of error-free maximally entangled states ahead of time. A $q$ ary EAQECC, denoted by $[[n,\kappa,\delta; c]]_q$, encodes $\kappa$ logical qudits into $n$ physical qudits, with the help of $n-\kappa-c$ ancillas and $c$ pairs of maximally entangled qudits. Such a quantum code can correct up to $\lfloor(\delta-1)/2\rfloor$ quantum errors. An EAQECC is a QECC if the code is designed without entanglement assistance, that is, when $c=0$.
	
With maximally entangled states as an additional resource, the pools of feasible classical ingredients in the construction of EAQECCs can include classical codes which are not self-orthogonal. A general construction of qubit EAQECCs was provided in \cite{Brun2006} via any binary or quaternary linear codes. This approach was generalized to the qudit case in \cite{Galindo2019}.
	
\begin{lem}{\rm \cite[Theorem 3]{Galindo2019}}\label{prop:two}
If $\mathcal{C}$ is an $[n,k,d]_{q^2}$ code, then there exists an $[[n,\kappa, \delta; c]]_q$ EAQECC $\mathcal{Q}$ with
\begin{align*}
c = k - \dim_{\mathbb{F}_{q^2}}
\left({\rm Hull}_{\rm H}(\cc)\right), \
\kappa = n-2k+c \mbox{, and}\
\delta = {\rm wt}\left(\mathcal{C}^{\perp_{\rm H}} \setminus {\rm Hull}_{\rm H}(\cc)\right).
\end{align*}
\end{lem}
	
If $\cc \subseteq \cc^{\perp_{\rm H}}$ or $\cc^{\perp_{\rm H}} \subseteq \cc$ in Lemma \ref{prop:two}, we obtain a QECC with parameters $[[n,n-2k, \delta_1]]_q$ or $[[n,2k-n, \delta_2]]_q$, where $\delta_1 = {\rm wt}\left(\mathcal{C}^{\perp_{\rm H}} \setminus \cc\right)$ and $\delta_2 = {\rm wt}\left(\mathcal{C} \setminus \cc^{\perp_{\rm H}}\right)$. The code $\mathcal{Q}$ in Lemma \ref{prop:two} is {\it nondegenerate} or {\it pure} if $\delta = d(\cc^{\perp_{\rm H}})$.
	
The Singleton-like bound for any $[[n,\kappa, \delta; c]]_q$ EAQECC in \cite[Corollary 9]{GraHubWin2022} reads
\begin{alignat}{5}
\kappa &\le c+\max\{0,n - 2 \delta + 2\},\label{eq:QMDS_small_distance}\\
\kappa &\le n-\delta+1,\label{eq:QMDS_trivial}\\
\kappa
&\le\frac{(n-\delta+1)(c+2\delta-2-n)}{3\delta-3-n} \mbox{, with }
\delta-1\ge\frac{n}{2}.\label{eq:QMDS_large_distance}
\end{alignat}
We call EAQECCs that achieve equality in \eqref{eq:QMDS_small_distance} whenever $\delta \le \frac{n}{2}$ and in \eqref{eq:QMDS_large_distance} whenever $\delta>\frac{n}{2}$ {\it maximum distance separable} (MDS) EAQECCs. Given a classical $[n,k,n-k+1]_{q^2}$ MDS code with $k<\lfloor n/2\rfloor$, whose Hermitian hull has dimension $\ell$, Lemma \ref{prop:two} produces a quantum MDS $[[n,n-k-\ell,k+1;k-\ell]]_q$ code . If $\ell=k$, then we get an $[[n,n-2k,k+1]]_q$ MDS QECC.
	
We use classical MDS codes whose Hermitian hulls are also MDS to construct quantum MDS codes. We show that these MDS EAQECCs beat the MDS QECCs in terms of minimum distance, for fixed length and dimension.

The following propagation rule for EAQECCs was established in \cite{Luo2021}.
\begin{lem}\label{lem:more}{\rm \cite[Theorem 12]{Luo2021}}
For $q>2$, the existence of a pure $[[n,\kappa,\delta;c]]_q$ code
$\mathcal{Q}$, constructed by Lemma \ref{prop:two}, implies the
existence of an $[[n, \kappa+i, \delta; c+i]]_q$ code
$\mathcal{Q}^{\prime}$ that is pure to distance $\delta$ for each $i
\in\{1,\ldots,\ell\}$, where $\ell$ is the dimension of the Hermitian
hull of the $\mathbb{F}_{q^2}$ linear code $\mathcal{C}$ that
corresponds to $\mathcal{Q}$.
\end{lem}
	
By Lemma \ref{lem:more}, an MDS EAQECC with large dimension can be constructed from an MDS QECC. In constructing quantum MDS codes, researchers have more recently focused on such quantum codes with large minimum distances. There remains significant open problems in this regard. Using MDS codes whose Hermitian hulls are MDS, we obtain MDS EAQECCs with large minimum distances for fixed length and dimension. By definition, the Hermitian hull of a linear code is a Hermitian self-orthogonal code. Given an $[n,k,n-k+1]_{q^2}$ MDS code whose Hermitian hull is an $[n,k-s,n-k+s+1]_q$ code with $k< \lfloor n/2 \rfloor$ and $1<s\leq \lfloor k/2\rfloor$, Lemma \ref{prop:two} generates an $[[n,n-2k+2s,k-s+1]]_q$ MDS QECC $\mathcal{Q}_1$ and an $[[n,n-2k+s,k+1;s]]_q$ MDS EAQECC $\mathcal{Q}_2$. Applying Lemma \ref{lem:more} on $\mathcal{Q}_2$, we obtain an $[[n,n-2k+2s,k+1;2s]]_q$ MDS EAQECC $\mathcal{Q}^{\prime}_2$. The quantum code $\mathcal{Q}^{\prime}_2$ outperforms $\mathcal{Q}_1$ in terms of minimum distance. It takes $2s$ pre-shared entangled states to improved the distance by $s$.
	
We provide an example based on the linear codes constructed based on Theorems \ref{GRScon1} and \ref{GRScon1-1}.
	
\begin{example}\label{example-QE-EAQE}
Theorem \ref{GRScon1} gives us a $[q^2,q,q^2-q+1]_q$ GRS code whose Hermitian hull is a $[q^2,q-1,q^2-q+2]_q$ code. By Lemma \ref{prop:two}, we obtain a $[[q^2,q^2-2q+2,q]]_q$ MDS QECC $\mathcal{Q}_1$ and a $[[q^2,q^2-2q+1,q+1;1]]_q$ MDS EAQECC $\mathcal{Q}_2$. Using Lemma \ref{lem:more}, the code $\mathcal{Q}_2$ gives rise to a $[[q^2,q^2-2q+2,q+1;2]]_q$ MDS EAQECC $\mathcal{Q}^{\prime}_2$. It is clear that $\mathcal{Q}^{\prime}_2$ and $\mathcal{Q}_1$ have the same length and dimension, whereas the minimum distance of $\mathcal{Q}^{\prime}_2$ is greater than that of $\mathcal{Q}_1$. A $[[q^2,q^2-2q+2,q]]_q$ MDS QECC is known from \cite{GRASSL2004}. What we have here is a code $\mathcal{Q}^{\prime}_2$ that has better error-correction capability with the use of maximally entangled qudits. We know from \cite{Ball2022} that the largest minimum distance of a $q$ ary MDS QECC of length $q^2$ is $q$. Theorem \ref{GRScon1-1} and Lemma \ref{prop:two} produce a $[[q^2,q^2-2(q+u)+1,q+u+1;1]]_q$ MDS EAQECC $\mathcal{Q}_3$ with $u<q-2$. Using some more maximally entangled qudits, Lemma \ref{lem:more} gives rise to a $[[q^2,q^2-2q+2,q+u+1;2u+2]]_q$ MDS EAQECC $\mathcal{Q}_3^{\prime}$.

The Hermitian hull of the GRS code constructed by Theorem \ref{GRScon1-1} is not MDS, but contains an MDS code as subcode. We cannot compare the MDS EAQECCs constructed by Theorem \ref{GRScon1-1} when $z>1$ with the $[[q^2,q^2-2q+2,q]]_q$ QECC, as the dimension of these MDS EAQECCs cannot be increased to $q^2-2q+2$. For instance, a $[[q^2,q^2-4q+2,2q+1;4]]_q$ MDS EAQECC is constructed by Theorem \ref{GRScon1-1} and Lemma \ref{prop:two}. Using Lemma \ref{lem:more}, the maximum dimension of the constructed MDS EAQECC is $q^2-2q-2$. What we highlight here is that not all MDS EAQECCs generated by MDS codes whose Hermitian hulls contain MDS codes are comparable with MDS QECCs in terms of minimum distance for fixed length and dimension.
\end{example}

Keeping the symbols from Example \ref{example-QE-EAQE}, Table \ref{table1} summarizes the parameters of MDS QECCs and EAQECCs that can be constructed via the linear codes proposed in Section \ref{sec:3}.

\begin{table*}[ht]
	\caption{Parameters of MDS QECCs and EAQECCs.}
	\label{table1}
	\renewcommand{\arraystretch}{1.2}
	\centering
	\begin{tabular}{c cc|ccc|c}
		\toprule
		No. & Length & Dimension & \multicolumn{3}{c|}{ Minimum Distance and Number of Pre-shared Pairs $(\delta;c)$ }& Constraints\\
		
		& $n$ & $\kappa$ &QECC $\mathcal{Q}_1$& EAQECC $\mathcal{Q}_2^{\prime}$& EAQECC $\mathcal{Q}_3^{\prime}$ & \\
		\midrule
		$1$ & $q^2$ &  $q^2-2q+2$ & $(q;0)$ & $(q+1;2)$ & $(q+u+1;2u+2)$ & $u<q-2$\\
		\midrule
		$2$ & $q^2-1$ &  $q^2-2k+1$ & $(k;0)$ & $(k+1;2)$ & $(k+u+1;2u+2)$ & $1<k<q$,\\
		& &&&&&$u<q-2$\\
		\midrule
		$3$ & $q^2-s(q+1)$ &  $q^2-s(q+1)$ & $(k;0)$ & $(k+1;2)$ & $(k+u+1;2u+2)$ & $1<k<q$,\\
		& & $ -2k+2$ &&&&$u<q-2$,\\
		& &&&&&$s=\gcd(k-1,q-1)$\\
		\midrule
	$4$ & $(q+1)(q-1-s)$ &  $(q+1)(q-1-s)$ & $(k;0)$ & $(k+1;2)$ & $(k+u+1;2u+2)$ & $1<k<q$,\\
	&& $ -2k+2$ &&&&$u<q-2$,\\
	&	&&&&&$s=\gcd(m-k+1,q-1)$\\
    \midrule
    $5$ & $n$ &  $n-2k+2$& $(k;0)$ & $(k+1;2)$ & - & $(n-1)$ divides $(q^2-1)$, \\
    &&&&&&$n\not=q^2$,\\
&&&&&&$1\leq k \leq \lfloor \frac{n-2}{q+1}\rfloor+2$\\
\midrule
$6$ & $tq$ &  $tq-2k+2$ & $(k;0)$ & $(k+1;2)$ & - & $1\le t\le q-1$,\\
&&&&&& $1\leq k \leq \lfloor \frac{tq-2}{q+1}\rfloor+2$\\
\midrule
$7$ & $(t+1)n_0+1$ &  $(t+1)n_0-2k+3$ & $(k;0)$ & $(k+1;2)$ & - & $1\leq k \leq \lfloor \frac{(t+1)n_0-1}{q+1}\rfloor+2$,\\
& &&&&& $n_0$ divides $(q^2-1)$,\\
& & & & & & $n_2=\frac{n_0}{\gcd(n_0,q+1)}$,\\
&&&&&& $1\leq t\leq \frac{q-1}{n_2}-2$ \\
\bottomrule
\end{tabular}
\end{table*}

Since Table \ref{table1} lists the parameters of all EAQECCs which can be compared with the corresponding MDS QECCs, we provide the parameters of the other EAQECCs constructed from linear codes in Section \ref{sec:3} in Table \ref{table2}. The pre-shared pairs do not come for free. The cost of sharing and purifying these pairs means that EAQECCs do not automatically outperform standard quantum codes in all circumstances. The number of pre-shared pairs for our MDS EAQECCs can be made flexible by using Lemma \ref{lem:more}.

Table \ref{table2} shows that our $q$ ary MDS EAQECCs have minimum distances greater than $q$. Most known $q$ ary MDS EAQECCs constructed from GRS codes have minimum distances which are less than $q$. Prior to our present work, there has been only one known construction of $q$ ary MDS QECCs with minimum distance greater than $q$. This class of $[[q^2+1,q^2-2q+1,q+1]]_q$ MDS QECCs was constructed in \cite{Guardia2011} from constacyclic codes, making their lengths rigid. The codes in Tables \ref{table1} and \ref{table2} can produce many MDS EAQECCs with new parameters. 

There has been accelerating effort to implement quantum information processing on qudit systems with $q>3$, beyond qubit or qutrit. Based on recent results reported in \cite{Ringbauer2022}, implementations related to some real application have been attempted for $q$-ary quantum codes when $q \leq 7$.

To control a single qudit, for $q>2$, is significantly more challenging than controlling a qubit. The former is much harder to calibrate. We do hope to see more of the challenges in building a qudit processor being tackled and overcome. For now, it appears that scaling the Hilbert space for qudit with large $q$ is not worth the explosion of complexity in contrast to keeping $q=2$ while exploring the vista for $3 \leq q \leq 7$. Table \ref{table3} lists the parameters, with minimum distance greater than $7$, of previously known MDS EAQECCs and the new ones over $\Ff_7$.

\begin{table*}[ht]
	\caption{Parameters of MDS EAQECCs with Minimum Distance $\delta = k+1$ and $q>2$.}
	\label{table2}
	\renewcommand{\arraystretch}{1.2}
	\centering
	\begin{tabular}{ccccc}
		\toprule
		No. & Length $n$ & Dimension $\kappa$ & Pre-shared Pairs $c$ & Constraints\\
		\midrule
		$1$ & $q^2$ &  $q^2-2k+z^2$ & $z^2$ & $1\leq z<\lfloor q/2\rfloor$,\\
		& &&& $zq\leq k<(z+1)q-z-1$\\
	\midrule
		$2$ & $q^2-1$ &  $q^2-1-2k+z^2$ & $z^2$ & $1\leq z<\lfloor \frac{q^2-1}{2q}\rfloor$,\\
		&&&& $z+f+1<q$,\\
		&&&& $zq\leq k<(z+1)q-z-f-1$\\
	\midrule
		$3$ & $q^2-s(q+1)$ &  $q^2-s(q+1)$ & $2z^2$ & $1\leq z<\lfloor \frac{q^2-s(q+1)}{2q}\rfloor$,\\
		&&$ -2k+2z^2$&& $z+f+1<q$, $f\geq z$,\\
		&&&& $s=\gcd(q-f-1,q-1)$,\\
		&&&& $zq\leq k<(z+1)q-z-f-1$\\
	\midrule
		$4$ & $q^2-s(q+1)$ &  $q^2-s(q+1)$ & $z^2+zf$ & $1\leq z<\lfloor \frac{q^2-s(q+1)}{2q}\rfloor$,\\
		&&$ -2k+z^2+zf$&& $z+f+1<q$, $1\leq f<z$,\\
		&&&& $s=\gcd(q-f-1,q-1)$,\\
		&&&& $zq\leq k<(z+1)q-z-f-1$\\
\midrule
		$5$ & $(q+1)(q-1-s)$ &  $(q+1)(q-1-s)$ & $2z^2$ & $1\leq z<\lfloor \frac{(q+1)(q-1-s)}{2q}\rfloor$,\\
		&&$ -2k+2z^2$&& $z+f+1<q$, $f\geq z$,\\
		&&&& $q-f-1<m<q-1$\\
		&&&& $s=\gcd(m-q+f+1,q-1)$,\\
		&&&& $zq\leq k<(z+1)q-z-f-1$\\
	\midrule
		$6$ & $(q+1)(q-1-s)$ &  $(q+1)(q-1-s)$ & $z^2+zf$ & $1\leq z<\lfloor \frac{(q+1)(q-1-s)}{2q}\rfloor$,\\
		&&$ -2k+z^2+zf$&& $z+f+1<q$, $1\leq f<z$,\\
		&&&& $q-f-1<m<q-1$\\
		&&&& $s=\gcd(m-q+f+1,q-1)$,\\
		&&&& $zq\leq k<(z+1)q-z-f-1$\\
\midrule
$7$ &$n$ &  $n-2k+c$ & $2\le c\le k$ & $(n-1)$ divides ${q^2-1}, \quad n \not=q^2$,\\
&&&& $1\leq k \leq \lfloor \frac{n-2}{q+1}\rfloor+2$\\
\midrule
$8$ &$tq$ &  $tq-2k+c$ & $2\le c\le k$ & $1\le t\le q-1, \quad 1\leq k \leq \lfloor \frac{tq-2}{q+1}\rfloor+2$,\\
\midrule
$9$ & $(t+1)n_0+1$ &  $(t+1)n_0+1-2k+c$ & $2\le c\le k$ & $1\leq k \leq \lfloor \frac{(t+1)n_0-1}{q+1}\rfloor+2$, $n_0$ divides $(q^2-1)$,\\
&& & & $n_2=\frac{n_0}{\gcd(n_0,q+1)}, \quad 1\leq t \leq\frac{q-1}{n_2}-2$ \\
\bottomrule
\end{tabular}
\end{table*}

\begin{table*}[ht]
	\caption{Parameters of MDS EAQECCs over $\Ff_7$ with minimum distance $\delta > 7$.}
	\label{table3}
	\renewcommand{\arraystretch}{1.2}
	\centering
	\begin{tabular}{lc | lc | lc | lc }
		\toprule
		$[[n,\kappa,\delta;c]]_7$ & Reference & $[[n,\kappa,\delta;c]]_7$ & Reference&
		$[[n,\kappa,\delta;c]]_7$ & Reference&
		$[[n,\kappa,\delta;c]]_7$ & Reference\\
		\midrule
		$[[50,36,8;0]]_7$ & \cite{Guardia2011} &
		
		$[[50,42,9;8]]_7$ &\cite{sari2021} &
		
		$[[50,41,10;9]]_7$ &\cite{sari2021} &
		
		$[[50,40,11;10]]_7$ &\cite{sari2021} \\
		
		$[[50,39,12;11]]_7$ &\cite{sari2021} &
		
		$[[50,38,13;12]]_7$ &\cite{sari2021} &
		
		$[[50,37,14;13]]_7$ &\cite{sari2021} &
		
		$[[50,36,15;14]]_7$ &\cite{sari2021} \\
		
		$[[50,35,16;15]]_7$ &\cite{sari2021} &
		
		$[[50,34,17;16]]_7$ &\cite{sari2021} &
		
		$[[50,33,18;17]]_7$ &\cite{sari2021} &
		
		$[[50,32,19;18]]_7$ &\cite{sari2021} \\
		
		$[[50,31,20;19]]_7$ &\cite{sari2021} &
		
		$[[50,30,21;20]]_7$ &\cite{sari2021} &
		
		$[[50,29,22;21]]_7$ &\cite{sari2021} &
		
		$[[50,28,23;22]]_7$ &\cite{sari2021} \\
		
		$[[50,27,24;23]]_7$ &\cite{sari2021} &
		
		$[[50,26,25;24]]_7$ &\cite{sari2021} &
		
		$[[49,36,8;1]]_7$ &\cite{Fan2016} &
		
		$[[49,34,9;1]]_7$ &\cite{Fan2016} \\
		
		$[[49,32,10;1]]_7$ &\cite{Fan2016} &
		
		$[[49,30,11;1]]_7$ &\cite{Fan2016} &
		
		$[[49,28,12;1]]_7$ &\cite{Fan2016} &
		
		$[[49,26,13;1]]_7$ &\cite{Fan2016} \\
		
		$[[25,18,8;7]]_7$ &\cite{Fan2016} &
		
		$[[25,17,9;8]]_7$ &\cite{Fan2016} &
		
		$[[25,16,10;9]]_7$ &\cite{Fan2016} &
		
		$[[25,15,11;10]]_7$ &\cite{Fan2016} \\
		
		$[[25,14,12;11]]_7$ &\cite{Fan2016} &
		
		$[[25,13,13;12]]_7$ &\cite{Fan2016} &
		
		$[[25,13,9;4]]_7$ &\cite{Wang2020} &
		
		$[[25,9,11;4]]_7$ &\cite{Wang2020} \\
		
		$[[25,5,13;4]]_7$ &\cite{Wang2020} &
		
		$[[24,12,8;2]]_7$ &\cite{Fan2016} &
		
		$[[24,10,9;2]]_7$ &\cite{Fan2016} &
		
		$[[24,8,10;2]]_7$ &\cite{Fan2016} \\
		
		$[[24,6,12;4]]_7$ &\cite{Chen2017} &
		
		$[[24,4,13;4]]_7$ &\cite{Chen2017} &
		
		$[[49,25,15;4]]_7$ & New &
		
		$[[49,23,16;4]]_7$ & New \\
		
		$[[49,21,17;4]]_7$ & New &
		
		$[[49,19,18;4]]_7$ & New &
		
		$[[49,16,22;9]]_7$ & New &
		
		$[[49,14,23;9]]_7$ & New \\
		
		$[[49,12,24;9]]_7$ & New &
		
		$[[41,29,8;2]]_7$ & New &
		
		$[[41,27,9;2]]_7$ & New &
		
		$[[41,25,10;2]]_7$ & New \\
		
		$[[41,23,11;2]]_7$ & New &
		
		$[[41,19,15;6]]_7$ & New &
		
		$[[41,17,16;6]]_7$ & New &
		
		$[[41,15,17;6]]_7$ & New \\
		
		$[[33,21,8;2]]_7$ & New &
		
		$[[33,19,9;2]]_7$ & New &
		
		$[[33,17,10;2]]_7$ & New &
		
		$[[33,10,16;8]]_7$ & New \\
		
		$[[25,13,8;2]]_7$ & New &
		
		$[[25,11,9;2]]_7$ & New &
		
		&  &
		
		&  \\
		
		\bottomrule
	\end{tabular}
\end{table*}

\section{Concluding remarks}\label{sec:5}

In this paper, we have studied Hermitian hulls of linear codes from cyclic codes and related GRS codes as well as from two-point rational AG codes. We have explored linear codes whose Hermitian hulls are MDS and construct a number of families of linear codes with such Hermitian hulls. We identify the following directions as interesting and worthy of further studies.
\begin{itemize}
    \item Construct, from cyclic codes and related GRS codes, new families of linear codes whose Hermitian hulls are MDS. The focus can be on codes with new parameters as their lengths and dimensions vary or on codes of the same length but having larger dimensions.
    \item Study the Hermitian hull of an AG code in a more general setting. One can, for instance, look for ways to enlarge the dimensions of two-point rational AG codes.
    \item Discover new families of two-point rational AG codes with different lengths and dimensions or with the same length but larger dimension.
    \item Construct AG codes, whose Hermitian hulls can be fully characterized, from more general algebraic curves.
    \item Codes under the Galois inner product have been studied in \cite{Cao2021, Islam2023a}. Codes whose hulls are MDS under the Galois inner product, beyond the Euclidean and Hermitian cases, may have some potential application in coding theory, including in quantum error-control. The claim stated as \cite[Corollary 2.5]{Liu2019}, however, must be firmly established before one can proceed to link the parameters of Galois hulls directly to those of EAQECCs. The classical codes must first be shown to correspond to a set of quantum error operators that the resulting quantum codes can handle properly.
\end{itemize}


\end{document}